\documentclass[titlepage,12pt]{article}
\usepackage[dvips]{graphicx}

\usepackage[myheadings]{fullpage}
\usepackage{pmetrika}
\usepackage{adjustbox}

\usepackage{amssymb}
\usepackage{amsmath}
\usepackage{graphicx,color}
\usepackage{amsfonts}
\usepackage[ignoreall]{geometry}
\usepackage{natbib}
\usepackage{enumitem}
\usepackage{setspace}
\usepackage{relsize}
\usepackage{graphicx}
\usepackage[document]{ragged2e}
\usepackage{capt-of}
\usepackage{url}

\usepackage{blkarray}

\usepackage{multirow}
\usepackage{longtable}

\begin{document}

%
%
%
%
%
%

\setcounter{page}{1}
\vspace*{2\baselineskip}

\title{Latent Theme Dictionary Model for Finding Co-occurrent Patterns in Process Data}\vskip3pt

\author{Guanhua Fang and Zhiliang  Ying \\
Columbia University}

\linespacing{1.5}
\abstracthead
\begin{abstract}

\justify
Process data, which are temporally ordered sequences of categorical observations, are of recent interest due to its increasing abundance and 
the desire to extract useful information. 
A process is a collection of time-stamped events of different types, recording how an individual behaves in a given time period. 
The process data are too complex in terms of size and irregularity for the classical psychometric models to be directly applicable and, consequently, new ways for modeling and analysis are desired. 
We introduce herein a latent theme dictionary model for processes
that identifies co-occurrent event patterns and individuals with similar behavioral patterns. Theoretical properties are established under certain regularity conditions for the likelihood-based estimation and inference.
A nonparametric Bayes algorithm using the Markov Chain Monte Carlo method is proposed for computation. Simulation studies show that the proposed approach performs well in a range of situations. The proposed method is applied to an item in the 2012 Programme for International Student Assessment with interpretable findings.

\end{abstract}\vspace{\fill}\pagebreak


\justify

\section{Introduction}\label{sec:1}

\textit{Process data} are temporally ordered data with categorical observations. Such data are ubiquitous and common in e-commerce (online purchases), social networking services and computer-based educational assessments. 
In large scale computer-based tests, 
analyzing process data has gained much attention and becomes a core task in the next generation of assessment; see, for example, 2012 and 2015 Programme for International Student Assessment 
\citep[PISA;][]{pisa2012, pisa2015}, 2012 Programme for International Assessment of Adult Competencies \citep[PIAAC;][]{goodman2013literacy}, Assessment and Teaching of 21st Century Skills \citep[ATC21S;][]{griffin2012assessment}.
In such technology-rich tests, there are problem-solving items which require the examinee to perform a number of actions before submitting final answers. 
These actions and their corresponding times are sequentially recorded and saved in a log file. Such log file data could provide extra information about the examinee's latent structure that is not available to traditional paper-based tests, in which only final responses (correct / incorrect) are collected.

Similar to item response theory \citep[IRT;][]{lord1980applications} models and diagnostic classification models \citep[DCMs;][]{templin2010diagnostic}, it is important to characterize item and examinees' characteristics through the calibration of item and person parameters in the analysis of process data.
However, process data are much more complicated in the sense that events occur at irregular time points and event sequence length varies from one examinee to another. 
Different examinees may have different reaction speeds in addition to varied action patterns to complete the task.
In addition, different examinees may have different strategies to reach their answers. These different behavioral patterns inherent in the process data allow us to classify examinees into different groups with meaningful interpretations.
Because some event sequences appear frequently, we may use sequential co-occurrent event patterns to extract important features from process data.



There is a recent literature on analysis of process data using data-mining tools. 
\cite{he2016analyzing} proposed to extract and detect robust sequential action patterns via n-gram method on problem-solving items in PIAAC.
\citet{qiao2018data} applied six different classification methods to a ``Tickets" item in PISA 2012 and compared their performances in terms of better feature selection.
\cite{han2019predictive} used a tree-based ensemble method to generate predictive features in PISA items.
These methods can extract useful features and predict individual performance. However, unlike classical latent variable-based psychometric models, they are essentially data mining algorithms. In particular, they are not generative and lack of statistical interpretation. Furthermore, they do not use the time stamps of actions, which are collected in the process data. On the other hand, model-based approaches to process data have also been developed in recent years. \cite{xu2018latent} proposed a Poisson process-based latent class model for clustering analysis. \cite{xu2019latent} developed a latent topic model with a Markovian structure for finding the underlying dimensions of examinees' latent ability. 
\cite{chen2019continuous} introduced a continuous-time dynamic choice model to characterize the decision making process.
Despite these efforts, statistical modeling of process data is still in its infancy and it is desirable to develop comprehensive methods that can systematically explore process data, especially in terms of simultaneously classifying individuals and extracting event features. 

This paper proposes a latent theme dictionary model (LTDM). 
Different from latent Dirichlet allocation \citep[LDA;][]{blei2003latent}, a well-known method for identifying word topic (semantic structure),
the proposed model is a latent class-type model with two layers of latent structure that assumes an underlying latent class structure for examinees and a latent ordered pattern association structure for event types (i.e. some event types may appear together frequently). 
To incorporate the temporal nature, a survival time model with intensity based on personal latent class is used for gap times between two consecutive events. 
The challenging issues of model identifiability are dealt with through using special dictionary structure and Kruskal's fundamental result of unique decomposition for three-dimensional arrays \citep{kruskal1977three}.
A nonparametric Bayes algorithm (NB-LTDM) is proposed to construct pattern dictionary, classify individuals, and estimate model parameters simultaneously.

The rest of paper is organized as follows. In Section \ref{sec:2}, we describe process data and ```Traffic" item from PISA 2012.
In Section \ref{sec:model}, we propose a new latent theme dictionary model, which combines LCM and TDM and incorporates time structure. In Section \ref{sec:3}, we develop theoretical results on model identifiability and estimation consistency. In Section \ref{sec:4}, we discuss the computational issue and propose the NB-LTDM algorithm.
The simulation results are presented in Section \ref{sec:5}. In Section \ref{sec:6}, we apply the proposed method to the ``Traffic" item in PISA 2012 and obtain some interpretable results. Some concluding remarks are given in Section \ref{sec:7}.

\section{\label{sec:2} Process Data and Traffic Item}


The process data here refer to a sequence of ordered events (actions) coupled with time stamps. For an examinee, his/her observed data are denoted by $((e_1,t_1), \ldots, (e_n, t_n)$, $\ldots, (e_N, t_N))$, where $e_n$ is the $n$th event and $t_n$ is its corresponding time stamp. 
We have $0 < t_1 < \ldots < t_N$ and $e_n \in \mathcal E$,
where $\mathcal E$ is the set of all possible event types. For notational simplicity, we write $e_{1:N} = (e_n: n = 1, \ldots, N)$ and $t_{1:N} = (t_n: n = 1, \ldots, N)$.

\begin{figure}
	\centering
	\includegraphics[width=6in]{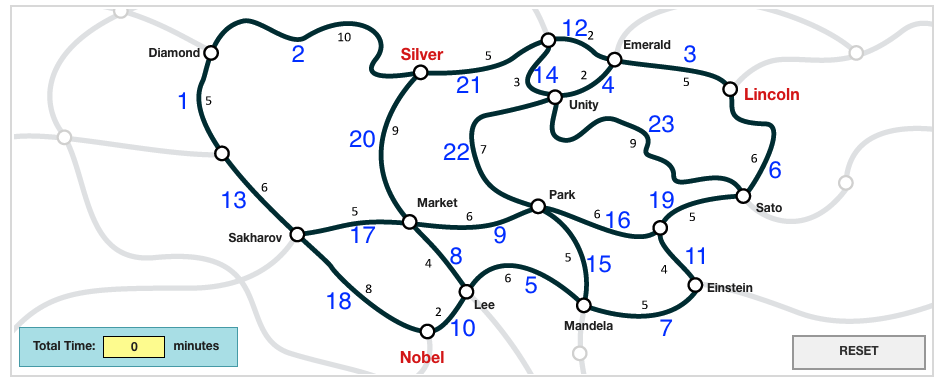}\\
	\caption{Map for ``Traffic" item interface, the big blue number is the label for road and the small black number represents the time for traveling on that road.}\label{map}
\end{figure}

We use the ``Traffic" item from PISA 2012 \citep{pisa2012} as a motivational example to illustrate various concepts and notation. 
This item is publicly available online at ``http://www.oecd.org/pisa/test-2012/testquestions/question1/".
PISA is a worldwide assessment to evaluate educational performances of different countries and economies. 
The ``Traffic" item contains three questions where the most challenging one asks the examinee to operate on a computer to complete the task, i.e., to locate a meeting point which is within 15 minutes away from three places, Silver, Lincoln and Nobel.
There are two correct answers, ``Park" and ``Silver" for this task.
Figure \ref{map} shows the initial state of the computer screen. There are 16 destinations and 23 roads in the map.
The integer in blue next to each road is the road number and the integer in black is the traveling time from one end to the other.
The examinee could click a road to highlight it, re-click a clicked road to unhighlight, and use ``RESET" button to remove all highlighted roads.
The ``Total Time" box shows the time for traveling on the highlighted roads.
Once a road is clicked, the corresponding time would be added to this box.
Each action and its corresponding time are sequentially saved in the log file during the process of completing the item.
A typical example of the action process of one specific examinee and its cleaned version are shown in Tables \ref{datastructure2} and \ref{cleandata}.
In this case, $\mathcal E = \{1, 2, \ldots, 23\}$. After removing unneeded rows (``START\_ITEM", ``END\_ITEM", ``Click", ``SELECT"),  we can see that there are 16 meaningful actions performed by this examinee as listed in Table \ref{cleandata}. His/her observed data are 
$$e_{1:16} = (10, 8, 9, \ldots, 9, 8, 10),
\quad t_{1:16} = (27.70, 28.60, 29.40, \ldots, 46.00, 47.70, 48.70).$$ 

\begin{table}
	\caption{The log file of an examinee.}
	\centering
	\label{datastructure2}
	{\scriptsize
		\begin{tabular}{l l l l}
			\hline\hline
			event\_number & event & time & event\_value \\
			\hline
			1 & START\_ITEM & 0.00 & NULL  \\
			2 & click & 24.60 & paragraph01  \\
			3 & ACER\_EVENT & 27.70 & 00000000010000000000000  \\
			4 & click & 27.70 & hit\_NobelLee  \\
			5 & ACER\_EVENT & 28.60 & 00000001010000000000000  \\
			6 & click & 28.60 & hit\_MarketLee  \\
			7 & ACER\_EVENT & 29.40 & 00000001110000000000000 \\
			8 & click & 29.40 & hit\_MarketPark  \\
			9 & ACER\_EVENT & 30.50 & 00000001110000000001000  \\
			... & ... & ... & ...  \\
			29 & ACER\_EVENT & 46.00 & 00110000100000000001010  \\
			30 & click & 46.00 & hit\_MarketPark  \\
			31 & ACER\_EVENT & 47.70 & 00110001100000000001010  \\
			32 & click & 47.70 & hit\_MarketLee \\
			33 & ACER\_EVENT & 48.70 & 00110001110000000001010 \\
			34 & click & 48.70 & hit\_NobelLee \\
			35 & Q3\_SELECT & 54.70 & Park  \\
			36 & END\_ITEM & 66.20 & NULL  \\
			\hline
		\end{tabular}
	}
\end{table}

\begin{table}
	\caption{The cleaned version of log data.}
	\centering
	\label{cleandata}
	\begin{tabular}{l l c}
		\hline\hline
		event\_number & time & event type\\
		\hline
		1 & 27.70 & 10 \\
		2 & 28.60 & 8 \\
		3 & 29.40 & 9 \\
		... & ... & ...  \\
		14 & 46.00 & 9 \\
		15 & 47.70 & 8 \\
		16 & 48.70 & 10 \\
		\hline
	\end{tabular}
\end{table}

As seen in the above example, process data are more complicated and also more informative compared with the classical item response data.
Different examinees may solve the item using different strategies and with different speeds that can only be seen from the response processes/process data, not the final answers/responses.
The form of process data is nonstandard in that action sequences for different examinees are not synchronized and have different lengths.
By extracting examinees' event patterns, including event co-occurrence and time heterogeneity, we can learn their problem-solving strategies and, consequently, better understand the underlying complex problem-solving (CPS) item. 
With these in mind, we introduce our new model in the next section.



\section{Latent Theme Dictionary Model}\label{sec:model}

In this section, we propose a latent theme dictionary model (LTDM).
We treat the whole event process of an examinee as a sequence of sentences where each sentence is an ordered subsequence of events.
Our proposed model focuses on modeling the event relationships within a sentence.
By doing this, we effectively reduce raw data length by splitting the original long sequence to multiple shorter sentences. This way of complexity reduction enables us to model sentences instead of the whole process which is more complicated.
We also want to point out that how to split the original event sequence to the sequence of event sentences is case-dependent which can be determined by the expert knowledge. Some special events (e.g. ``Reset" action) can be used to split event sequences in general.

To be precise, we assume that $e_{1:N}$ and $t_{1:N}$ are divided into sentence sequences, i.e. $e_{1:N} = (E_1, \ldots, E_k, \ldots, E_{K})$ and $t_{1:N} = (T_1, \ldots, T_k, \ldots, T_{K})$, where 
$$E_k = (e_{k,1}, \ldots, e_{k,u}, \ldots, e_{k, n_k}) ~ \textrm{and} ~ T_k = (t_{k,1}, \ldots, t_{k,u}, \ldots, t_{k, n_k})$$ 
are called event sentence and time sentence, respectively.
We use $w = [e_1 ~ \ldots ~ e_{l_w}]$ to represent an ordered pattern that events $e_1, \ldots, e_{l_w}$ appear sequentially and call it $l_w$-\textit{gram} (length of this pattern is $l_w$).
Because a 1-gram pattern is also an event, $e$ will be used both for event and 1-gram pattern throughout the sequel.  
Event sentence $E_k$ can be represented as a sequence of patterns, 
$$E_k = (w_{k,1}, \ldots, w_{k,u}, \ldots, w_{k,l_{E_k}}),$$
where $l_{E_k}$ is the number of patterns that $E_k$ contains.
Note that an event sentence can be partitioned into different pattern sequences. We use $\mathcal F(E)$ to denote the set of all possible pattern separations for $E$. 
We let $M_l$ be the number of different event patterns of length $l$. 
We define pattern dictionary $\mathcal D$ as the set of all distinct patterns and use $v_D$ to denote its cardinality, i.e. the size of the dictionary. Obviously, $v_D = M_1 + \ldots + M_l + \ldots M_L$, where $L$ is the maximum length of patterns.
Next we use ``Traffic" item as an example to illustrate the relation between event sentences and event separations.  
\begin{example}
	Consider an examinee taking the following actions to complete the ``Traffic" item,
	$$e_{1:16} = (10, 8, 9, 20, 3, 22, 4, 9, 8, 10, 16, 19, 6, 9, 8, 10).$$
	We split this observed event sequence to several sentences by using the following criteria. We treat a sentence as a subsequence of actions of consecutively highlighting/unhighlighting the road. In this case, we have a total of three sentences,
	\begin{eqnarray*}
		E_1 &=&  (10, 8, 9, 20, 3, 22, 4) \\
		E_2 &=&  (9, 8, 10) \\
		E_3 &=& (16, 19, 6, 9, 8,10).
	\end{eqnarray*}
	Suppose the underlying dictionary $\mathcal D$ is 
	$$\{[10], [20],[8~9], [9~8], [10~8], [9~20], [3~22~4],[9~8~10], [16~19~6]\}.$$
	Then, according to this dictionary, we have
	\begin{eqnarray*}
		& &\mathcal F(E_1) = \{S_{1,1}, S_{1,2}\}~ \textrm{with} ~ S_{1,1} = ([10],[8~9],[20],[3~22~4]) 
		~\textrm{and}~ S_{1,2} = ([10~8], [9~20], [3~22~4]), \\
		& &\mathcal F(E_2) = \{S_{2,1}, S_{2,2}\}~ \textrm{with} ~ S_{2,1} = ([9~8],[10]) 
		~\textrm{and}~ S_{2,2} = ([9~8~10]), ~ \textrm{and} \\
		& &\mathcal F(E_3) = \{S_{3,1}, S_{3,2}\}~ \textrm{with} ~ S_{3,1} = ([16~19~6], [9~8],[10]) 
		~\textrm{and}~ S_{3,2} = ([16~19~6],[9~8~10]).
	\end{eqnarray*}
\end{example}

We now specify the probability structure of our proposed model.
On a very high level, the proposed model is motivated by two simpler models, latent class model \citep[LCM;][]{gibson1959three} and theme dictionary model \citep[TDM;][]{deng2014association}; see Appendix E for their definitions. Suppose that the entire population consists of $J$ different classes of examinees, but their class labels are unknown. We use $z \in \{1, \ldots, J\}$ to denote the latent class to which the examinee belongs. 
Latent variable $z$ can be viewed as the examinee's latent attribute or discretized version of latent ability.
In order to jointly model $(E_1, \ldots, E_K)$ and $(T_1, \ldots, T_K)$, it is equivalent to model $(E_1, \ldots, E_K)$ and $(\tilde T_1, \ldots, \tilde T_K)$ where 
\begin{eqnarray}
\tilde T_k = (\tilde t_{k,1}, \ldots, \tilde t_{k,n_k}) ~ \textrm{for}~ k = 1, \ldots, K
\end{eqnarray}
with $\tilde t_{1,1} = t_{1,1}$, $\tilde t_{k,1} = t_{k,1} - t_{k-1, n_{k-1}}$ for $k \geq 2$, and $\tilde t_{k,u} = t_{k,u} - t_{k, u-1}$ for $u \geq 2$.
We make the usual local independence assumption, i.e., 
\begin{eqnarray}\label{cond_ind0}
P((E_1, \ldots, E_K), (\tilde T_1, \ldots, \tilde T_K) | z) 
&=& \prod_{k = 1}^K P(E_k, \tilde T_k | z).
\end{eqnarray}
This leads to  
\begin{eqnarray}\label{cond_ind}
P((E_1, \ldots, E_K), (\tilde T_1, \ldots, \tilde T_K)) 
&=& \sum_{z = 1}^J \pi_z \prod_{k = 1}^K P(E_k, \tilde T_k | z),
\end{eqnarray}
where $\pi_z$ is the probability mass for the $z$th latent class.
Thus the sentences are exchangeable. 

Next we model event sentence and time sentence by making use of the following conditional probability formula,
\begin{eqnarray}\label{prob_sep0}
P(E_k, \tilde T_k|z) = P(E_k | z) P(\tilde T_k | E_k, z).
\end{eqnarray}
For the event part, notice that each observed event sentence may have different pattern separations (i.e. two different pattern sequences can lead to the same event sequence). Thus
\begin{eqnarray}\label{prob_E}
P(E_k | z) = \sum_{S \in \mathcal F(E_k)} P(E_k|S, z) P(S|z).
\end{eqnarray}
%
%
We further assume 
\begin{eqnarray}\label{ass:S}
P(S | z) =  \frac{1}{n_S !} \prod_{w = 1}^{v_D} \theta_{zw}^{\mathbf 1_{\{w \in S\}}} (1 - \theta_{zw})^{\mathbf 1_{\{w \notin S\}}}, 
\end{eqnarray}
where $\theta_{zw} = P(w \in S | z)$ and $n_{S}$ is the number of patterns in $S$. Here \eqref{ass:S} is analogous to \eqref{eq:TDM} in the TDM setting. It specifies that an examinee from latent class $z$ has event pattern $w$ with probability $\theta_{zw}$. The extra term $\frac{1}{n_S!}$ comes from the fact that we consider the pattern orders.

For modeling $\tilde T_k$, we assume 
\begin{eqnarray}\label{ass:time}
P(\tilde T_k | E_k, z) = \prod_{u=1}^{n_k} P(\tilde t_{k,u}| z),
\end{eqnarray}
where $n_k$ is the length of $E_k$. In other words, the gap time between two consecutive events is assumed to be stationary given the latent class label. 

From \eqref{prob_sep0} - \eqref{ass:time} and the fact that $P(E_k | S, z) \equiv \mathbf 1_{\{S \in \mathcal F(E_k)\}}$, we have 
\begin{eqnarray}\label{prob_sent}
P(E_k, \tilde T_k | z)
& = &  \bigg\{\sum_{S \in \mathcal F(E_k)} \big[ \frac{1}{n_S !} \prod_{w = 1}^{v_D} \theta_{zw}^{\mathbf 1_{\{w \in S\}}} (1 - \theta_{zw})^{\mathbf 1_{\{w \notin S\}}} \big] \bigg\} \big[
\prod_{u=1}^{n_k} P(\tilde t_{k,u} | z) \big].
\end{eqnarray}
Finally, we specify that gap time $\tilde t_{k,u}$ follows an exponential distribution, i.e.,
\begin{eqnarray}
P(\tilde t_{k,u}|z) = \lambda_z \exp\{- \lambda_z \tilde t_{k,u}\}.
\end{eqnarray}
Different from traditional response time model \citep{van2006lognormal} where only the total time spent on an item is considered, the proposed model accounts for the time allocated to each action. 
Here gap time is assumed to follow the exponential distribution while the log-normal assumption is made in classical response time model.
Variable $\lambda_z$ can be viewed as the personal intensity. 
Such modeling often appears in the literature of event history analysis and survival analysis \citep{allison1984event, aalen2008survival}. 
Here we assume that $\lambda_z$ only depends on the latent class label. 
In general, it could be individual-specific which is related to the frailty model \citep{duchateau2007frailty}. It could also be event-dependent which is known as the competing risk analysis in the survival analysis. \cite{chen2019continuous}'s dynamic choice model considers the intensity function with both individual effect and event task effect.

Lastly, we assume that $K$, which is 
the number of sentences for the subject, follows some distribution function $F$ supported on $\mathbb Z = \{0,1,2, \ldots\}$, i.e.,
\begin{eqnarray}\label{number_sent}
P(K \leq k) =  F(k), ~~ k = 0, 1, 2, \ldots.
\end{eqnarray}
For simplicity, we may assume $F$ is a cumulative distribution function of Poisson random variable with parameter $\kappa$.
Note that instead of directly modeling $K$, we can also fit the proposed model conditioning on $K_i$s. 
This will not effect the estimates of other model parameters.

To summarize, LTDM is a data-driven model  that can be used for learning event patterns and population clustering simultaneously. 
The model framework is built on a very general level in the sense that (1) by letting number of latent classes be 1, the model reduces to TDM-type model with an ordered pattern dictionary and
(2) it reduces to a model for event sentences only if we let $\lambda_j$ be constant across all latent classes (i.e. $\lambda_j = \lambda, j = 1, \ldots, J$).

One important set of parameters, $\{\theta_{jw}\}:=\{\theta_{jw}; j=1,\ldots,J, w=1,\ldots,v_D\}$, measures how often examinees from different classes use distinct patterns.   
Another set of parameters, $\{\lambda_j\} := \{\lambda_j\}_{j=1}^J$, measures how fast examinees take actions across different groups.
In other words, the population are stratified by two factors, event pattern and respond speed.
From the information theory viewpoint, we can write 
\begin{eqnarray}\label{fisherE}
\mathcal I_{e_{1:N}, t_{1:N}}(\theta) = \mathcal I_{e_{1:N}}(\theta) + \mathcal I_{t_{1:N} | e_{1:N}}(\theta)
\end{eqnarray} 
and 
\begin{eqnarray}\label{fisherT}
\mathcal I_{e_{1:N}, t_{1:N}}(\theta) = \mathcal I_{t_{1:N}}(\theta) + \mathcal I_{e_{1:N} | t_{1:N}}(\theta),
\end{eqnarray} 
where $\mathcal I_{A}(\theta)$ is the Fisher information with respect to a generic parameter $\theta$ for some random vector $A$ and $\mathcal I_{A|B}(\theta)$ is the conditional Fisher information of $\theta$ for some generic random vector $A$ given random vector $B$.
In particular, taking $\theta$ to be event pattern parameters, $\{\theta_{jw}\}$,
we know that $\mathcal I_{e_{1:N}, t_{1:N}}(\theta_{jw}) \geq \mathcal I_{e_{1:N}}(\theta_{jw})$.
According to Proposition \ref{prop:decomp} in Appendix B, the equality is achieved if and only if $\lambda_j = \lambda, j = 1,\ldots,J$. This implies that we can estimate event parameter more accurately by observing $t_{1:N}$ when response speed are different across different groups.
Therefore, the inclusion of event times does not only characterizes the response speed of examinees from different classes  but also improves the estimation accuracy of model parameters.


We now construct the likelihood function for LTDM. 
We use $m$ to denote the total number of examinees and subscript $i$ to denote the $i$-th examinee. Assume that examinees are independent of each other. Then 
the complete likelihood of $\{ e_{1:N_i}, t_{1:N_i}, K_i, \mathbf S_i, z_i\}_{i=1}^m$ has the following expression,
\begin{eqnarray*}
	L_m & = & 
	\prod_{i=1}^m \bigg \{ \big\{ \prod_{k = 1}^{K_i} \pi_{z_i} P(S_{ik}, \tilde T_{ik}|z_i) \mathbf 1_{\{S_{ik} \in \mathcal F(E_{ik})\}} \big\}
    \frac{\kappa^{K_i} \exp\{-\kappa\}}{K_i !} \bigg \},
\end{eqnarray*}
where $\mathbf S_i = \{S_{ik}, k = 1, \ldots, K\}$.
Furthermore, by summing over/integrating out the unobserved latent variables, we have
\begin{eqnarray}\label{likli}
& & P(\{e_{1:N_i}, t_{1:N_i}\}_{i=1}^m) \nonumber \\
& = & \prod_{i=1}^m \bigg\{ \frac{\kappa^{K_i} \exp\{-\kappa\}}{K_i!} \sum_{z_i=1}^J \pi_{z_i}\prod_{k=1}^{K_i}
\big\{\sum_{S_{ik} \in \mathcal F(E_{ik})} \frac{1}{n_{S_{ik}} !} \prod_{w = 1}^{v_D} \theta_{z_iw}^{\mathbf 1_{\{w \in S_{ik}\}}} (1 - \theta_{z_iw})^{\mathbf 1_{\{w \notin S_{ik}\}}} \nonumber \\
& & \times \prod_{u=1}^{n_{ik}} p(\tilde t_{ik,u} | z_i) \big\}  \bigg\}.
\end{eqnarray}

\section{Identifiability \label{sec:3}}

Latent class models (LCMs) often face the issue of identifiability. 
There is an existing literature on identifiability of latent variable models; see \cite{Allman}, \cite{xu2017identifiability}, and references therein.
Theme dictionary model also has the identifiability issue since the underlying separations are unobserved. 
In this section, we address the identifiability issue, specifically towards dictionary and model parameters in the proposed model.
Intuitively, the dictionary and model parameters can be identified if the examinees  from different classes have distinct behaviors. To be more specific, two examinees may have different speeds to solve the item and their strategies (event patterns) should be different. 
In the following, we mathematically investigate this problem and identify the conditions under which the model becomes identifiable.   

We use $\mathcal O$ to denote the set of all possible sentences generated by $\mathcal D$, that is,
$\mathcal O = \bigcup_{j = 1}^J \mathcal O_j$ where $\mathcal O_j$ is the set of sentences generated by the pattern set of Class $j$, $\mathcal D_j =\{w: \theta_{j w} \neq 0\}$.
We define the set
$$\mathfrak{P} \equiv \{(\mathcal D, \{\theta_{jw}\}, \{\lambda_j\}, \pi, \kappa)\big| \mathcal D \in \mathbb D, \theta_{jw} \in [0,1], \lambda_j \in \mathbb R^{+}, \pi \in \mathcal S_J^{+}, \kappa \in \mathbb R^{+} \},$$
where $\mathbb D = \{\mathcal D ~ | ~ \mathcal D \textrm{ satisifes  A1 and A2 given below}\}$ and $\mathcal S_J^{+} = \{\pi \mid \pi > 0 \textrm{ and } \|\pi\|_1 = 1 \}$. 
We use $\mathcal P$ to denote an LTDM which depends on $(\mathcal D, \{\theta_{jw}\}, \{\lambda_j\}, \pi, \kappa)$. 
In the sequel, we will omit $(\mathcal D, \{\theta_{jw}\}, \{\lambda_j\}, \pi, \kappa)$ and use $\mathcal P$ when there is no ambiguity.

We say classes $j_1$ and $j_2$ are equivalent if $\lambda_{j_1} = \lambda_{j_2}$. 
We define the set of equivalence classes as $[j] = \{j_1 ~ | ~ j_1 \in \{1, \ldots, J\}, ~  j_1 ~ \textrm{and}~ j~\textrm{are equivalent.}  \}$.
Let $\mathcal D_{[j]}$ be the pattern dictionary of equivalence class $[j]$ and $\mathcal O_{[j]}$ be the set of all possible sentences generated by $\mathcal D_{[j]}$. 

\begin{enumerate}
	\item[A1] 
	For any class $j$ and any event $e$,
	it holds that $E \in \mathcal O_{j}$ if $E \in \mathcal O_{[j]}$ contains a subsentence $E_1 \in \mathcal O_{j}$ with $n_{j,e}$ consecutive events in set $\mathcal E - \{e\}$. 
	Here $n_{j,e}$ is the length of longest sentence in $\mathcal O_{j}$ without $e$.
	\item[A2] For every pattern $w = [e_1 ~ e_2 \ldots ~ e_{l_w}]$ in dictionary $\mathcal D$, $e_1, \ldots, e_{l_w}$ are distinct if $l_w \geq 2$.
\end{enumerate}
Assumption A1 is a technical condition to ensure that the patterns from distinct classes could be identified. 
This assumption is satisfied automatically in many cases such as (1) entries of $\{\lambda_j\}$ are different or (2) $\theta_{jw} > 0$ for all $j, w$.
Assumption A2 essentially restricts the dictionary in such a way that each pattern consists of distinct events. Clearly it is very easy to check. It is also natural in the sense that we do not want to treat too many replicated events as a pattern. 
For example, if we have two dictionaries $\mathcal D_1 = \{A, AA, AAA, AAAAAA\}$ and $\mathcal D_2 = \{A, AA, AAAA, AAAAA\}$, then they would generate the same sentence set (i.e., $\mathcal O_1 = \mathcal O_2$). Thus Assumption $A_2$ is necessary for dictionary identifiability.
We now introduce the formal definition of identifiability for LTDM.

\begin{definition}
	We say $(\mathcal D^{\ast}, \{\theta_{jw}^{\ast}\}, \{\lambda_j^{\ast}\},\pi^{\ast},\kappa^{\ast}) \in \mathfrak{P}$, is identifiable,
	if for any $(\mathcal D^{'}, \{\theta_{jw}^{'}\}, \{\lambda_j^{'}\},\pi^{'},\kappa^{'}) \in \mathfrak{P}$ 
	that satisfies
	\begin{eqnarray*}
		P(e_{1:N}, t_{1:N} |\mathcal D^{'}, \{\theta_{jw}^{'}\}, \{\lambda_j^{'}\}, \pi^{'}, \kappa^{'}) 
		&=& P(e_{1:N}, t_{1:N} |\mathcal D^{\ast}, \{\theta_{jw}^{\ast}\}, \{\lambda_j^{\ast}\}, \pi^{\ast}, \kappa^{\ast}) 
 	\end{eqnarray*}
	for all $e_{1:N}$ and $t_{1:N}$, and
	$\mathcal O^{'} = \mathcal O^{\ast}$, $|\mathcal D^{'}| \leq |\mathcal D^{\ast}|$,
	we must have
	\begin{eqnarray*}
		\mathcal D^{'} = \mathcal D^{\ast}, \quad \kappa^{'} = \kappa^{\ast}, \quad \textrm{and} \quad
		(\{ \theta_{jw}^{'}\}, \{\lambda_j^{'}\}, \pi^{'}) \overset{p}{=} (\{\theta_{jw}^{\ast}\}, \{\lambda_j^{\ast}\}, \pi^{\ast}).
	\end{eqnarray*}
\end{definition}
Here we use superscript $\ast$ to denote the true model (parameters/dictionary). 
$A \overset{p}{=} B$ means $A$ equals $B$ up to a permutation of class labels. 

We want to point out that in general $\mathfrak P$ is too large to be identifiable without additional constraints.
It is thus worth specifying a restricted space $\mathfrak{P}^{0} \subset \mathfrak{P}$ such that every model dictionary and parameter in $\mathfrak{P}^{0}$ is identifiable. 
Given Conditions C1 and C2 as specified in Appendix A,
we define 
$$\mathfrak{P}^{0} := \{(\mathcal D, \{\theta_{jw}\}, \{\lambda_j\}, \pi, \kappa ) \big| (\mathcal D, \{\theta_{jw}\}, \{\lambda_j\}, \pi, \kappa ) \in \mathfrak P ~ \textrm{and satisifes C1, C2}\} \}.$$
We have the following theorem.
\begin{theorem}\label{main}
	Under Conditions C1 and C2, every $(\mathcal D, \{\theta_{jw}\}, \{\lambda_j\}, \pi, \kappa )$ in $\mathfrak{P}^{0}$ is identifiable.
\end{theorem}
One immediate result as stated in Corollary \ref{cor:1class} is that  there are no two distinct dictionaries which have the same $\mathcal O$ if they satisfy Assumption A2. 
This also serves as a sufficient condition for the identifiability of TDMs, since a TDM could be non-identifiable without any additional assumption.

\begin{corollary}\label{cor:1class}
	For the 1-class case, if $\mathcal D$ and $\mathcal D^{'}$ satisfy Condition A2, then $\mathcal O = \mathcal O^{'}$ if and only if $\mathcal D = \mathcal D^{'}$.
\end{corollary}


Suppose that the number of latent classes $J$ and dictionary size $v_D$ are known.
The true dictionary and model parameters can be estimated consistently. The results are stated as follow.
\begin{theorem}\label{consistent}
	Define the maximum likelihood estimator
	\begin{eqnarray*}
		(\hat{\mathcal D}, \{\hat \theta_{jw}\}, \{\hat \lambda_j \}, \hat \pi, \hat \kappa ) = \textrm{argmax}_{(\mathcal D,\{\theta_{jw}\}, \{\lambda_j\}, \pi, \kappa) \in \mathfrak P_c} \prod_{i=1}^m L(e_{1:N_i}, t_{1:N_i}).
	\end{eqnarray*}
	where 
	$\mathfrak P_c := \mathbb D_c \times \boldsymbol \Theta_c$; $\mathbb D_c \subset \mathbb D$ is the set of dictionaries with size smaller than $v_D$ and $\boldsymbol \Theta_c$ is any compact subset containing the true parameter vector.
	Then, under Conditions C1 and C2, we have that
	$$P (\hat{\mathcal D} = \mathcal D^{\ast}) \rightarrow 1 $$
	and, for some permutation function $\rho$,
	$$ P\bigg(|\hat \kappa - \kappa^{\ast}| < \delta, \| \rho(\hat \lambda_{j}) - \lambda_{j}^{\ast}\|_2 < \delta, \| \rho(\hat \theta_{jw}) - \theta_{jw}^{\ast}\|_2 < \delta, \|\rho(\hat \pi) - \pi^{\ast} \|_2 < \delta \bigg) \rightarrow 1$$
	for any $\delta > 0$ as $m \rightarrow \infty$.
\end{theorem}

We would like to point out that if we only consider the event sentences and ignore the times, the above results still hold since all classes are in a single equivalence class.

\section{Computation \label{sec:4}}

Although LTDM postulates a parametric form,
we do not know the size of the true dictionary ($v_D$) and the number of latent classes ($J$) in practice.
Therefore, three challenges remain in terms of computation, namely, (1) finding the true underlying patterns (construction of dictionary), (2) clustering people into the right groups, and (3) computational complexity. We propose a new nonparametric Bayes - LTDM (NB-LTDM) algorithm as described below to address these issues. 

\begin{itemize}
	\item[] \textbf{NB-LTDM Algorithm}
	\item[] \textbf{Initialization}:  Randomly choose a large $J$; sample personal latent labels $z_i$ from the uniform distribution on $\{1,\ldots,J\}$; sample parameters $\{\theta_{jw}\}$ uniformly on [0,1]; sample $\pi$ from the Dirichlet distribution; sample $\{\lambda_j\}$ and $\kappa$ from $\mathrm{exp}(1)$. The initial dictionary $\mathcal D^{(0)}$ should include all $M_1$ 1-grams and a random selection of $S_0$ $l$-grams ($l = 2, \ldots, L$).
	
	\item[] \textbf{Output}: $J^{\ast}$ - the number of classes, $\mathcal D^{\ast}$ - the dictionary, estimates of model parameters.
	
	\item[] The algorithm takes the following iterative steps until the Markov chain becomes stable.
	\begin{itemize}
		\item[1] [\textbf{Search}] Within each latent class, we calculate the frequency of $l$-grams based on count. We find the $S$ most frequent $l$-grams ($l = 2, \ldots, L$) which do not appear in the current dictionary and add them into $\mathcal D$.
		\item[2] \textbf{[Split]} Split the event sequences according to the current dictionary.
		\item[3] \textbf{[Sample]} Sample separation for each event sequence from the corresponding possible candidates.
		\item[4] \textbf{[Inner part]} Use slice Gibbs \citep{walker2007sampling}
		sampling schemes to iteratively update the following variables:
		\begin{itemize}
			\item[] Model parameters $\{\theta_{jw}\}, \{\lambda_j\}$ and $ \kappa$, augmented variables, separations $\{S_{ik}\}$, latent labels $\{z_i\}$ and the prior parameters.
		\end{itemize}
		\item[5] [\textbf{Trim dictionary}] For each action pattern $w$ in the current dictionary, calculate the \textit{evidence probability} $\beta_w = \max_{j} \theta_{jw}$. Discard those patterns with evidence probability smaller than $\tau$.
	\end{itemize}
\end{itemize}

We set threshold $\kappa_m = \frac{1}{\sqrt{m}}$ and estimate the number of latent classes by $J^{\ast} = \#\{h | \pi_h > \kappa_m, h = 1,2,\ldots \}$. The estimated pattern dictionary is $\mathcal D^{\ast} = \{w | w ~ \textrm{is in dictionary at least half of time in the last 100 iterations.}\}$.  We use posterior means for other parameters.

We comment on the tuning parameters in the proposed algorithm:
$\tau$ is a threshold to filter out less frequent patterns; 
$S_0$ is the number of $l$-grams ($l = 2, \ldots, L$) in the initial dictionary;
$S$ controls the number of new patterns added into current dictionary.
We found in our simulation studies that the proposed method is not sensitive to the choices of $S$ and $S_0$. In practice, we may choose $S = 2 M_1$ and $S_0 = M_1$.

This data-driven method consists of two main steps, updating the dictionary and updating the model parameters. 
\begin{itemize}
	\item[] \textbf{Update dictionary}: 
	In each loop, the algorithm trims the dictionary by keeping patterns with high evidence level and discarding those with weak signals.
	Then it finds patterns (2-grams, \ldots, $L$-grams) with high frequencies within each latent class and adds new patterns to the current dictionary.
	This step can be viewed as a forward-backward-type variable selection \citep{tibshirani1997lasso, borboudakis2019forward} technique for dictionary update.
	Compared with the full Bayesian methods \citep[e.g. spike-and-slab prior;][]{ishwaran2003detecting, ishwaran2005spike} for dictionary selection, it can result in substantial reduction in the computational time.
	\item[] \textbf{Update parameters}:
	To update the model parameters, we follow the approach of \cite{dunson2009nonparametric}.
	In our specific setting, for a given dictionary, the parameters are updated using a Markov Chain Monte Carlo (MCMC) method together with the slice sampler.
	It allows us to avoid directly computing the marginal likelihood that requires massive computation in terms of integration of latent variables $\{z_i\}_{i=1}^m$ and $\{\mathbf S_i\}_{i=1}^m$.
	We use a stick-breaking prior \citep{sethuraman1994constructive} on latent class probabilities to avoid specifying a priori number of classes $J$.
	For precise mathematical formulation and updating rules in the inner part, see Appendix C.
\end{itemize}
Here we want to point out that we are not able to develop theoretical results for convergence analysis. However, the proposed method performs well in our simulation studies.

\section{Simulation Studies \label{sec:5}}

The simulation studies include four different simulation settings, which are specified below.

\begin{enumerate}
	\item In the first simulation setting, dictionary $\mathcal D$ consists of $1$-grams, $2$-grams and $3$-grams, with details given in Table \ref{sim1}. We set $v_D = 50$, $L = 3$, $M_1 = 20, M_2 = 20, M_3 = 10$ and set $m = 1000$, $J = 5$. Other model parameters are set as follow, $\pi = (0.4, 0.3, 0.2, 0.05, 0.05)$, $\{\lambda_j\} = (10,2.5,1,0.5,0.2)$ and $\kappa = 10$.  Pattern probability $\{\theta_{jw}\}$ is provided in Table \ref{sim1}.
	
	\tablehere{3}
	
	Under this setting, it can be verified that the model dictionary and parameters are identifiable:
	A1 is satisfied since $\lambda_j$'s are different; 
	A2 holds by the construction of dictionary;
	C1 and C2 are satisfied as the size of each equivalence class is 1. 
	
	\item 
	In the second simulation setting, dictionary $\mathcal D$ consists of $1$-grams, $2$-grams and $3$-grams,
	with details given in Table \ref{sim2}. We set $v_D = 50$, $L = 3$, $M_1 = 20, M_2 = 20, M_3 = 10$ and set $m = 1000$, $J = 6$. Other model parameters are set as follows, $\pi = (0.2,0.2,0.2,0.2,0.1,0.1)$, $\{\lambda_j\} = (0.2,4,0.2,4,1,1)$ and $\kappa = 10$.  Pattern probability $\{\theta_{jw}\}$ is provided in Table \ref{sim2}.
	
	\tablehere{4}
	
	Under this setting, the model dictionary and parameters are identifiable. 
	It is easy to see that there are four equivalence classes, $[1]$, $[2]$, $[5]$ and $[6]$ where $[1] = \{1,3\}$ and $[2] = \{2,4\}$.
	Assumption A1 is satisfied by observing that pattern dictionaries of Class 1 (2) and Class 3 (4) do not overlap.
	Assumption A2 holds by the construction of dictionary.  
	We can construct a partition $\mathcal I_1 = \{1,2,\ldots,5\}, \mathcal I_2 = \{6, \ldots, 10\}, \mathcal I_3 = \{11, \ldots, 20\}$ such that sentence $(e_1, e_2, e_3)$ has only one separation for any $e_k \in \mathcal I_k$ ($k = 1, 2, 3$).
	It can be checked directly that the corresponding $T$-matrices have full column rank. Therefore C1 is satisfied.
    Condition C2 can also be verified similarly.
    
	\item
	In the third simulation setting, dictionary $\mathcal D$ includes patterns up to $4$-grams, with details provided in Table \ref{sim3}.
	We let $v_D = 90$, $L = 4$, $M_1 = 30, M_2 = 30, M_3 = 15, M_4 = 15$ and set $m = 2000$, $J = 5$. Other model parameters are set as follows, $\pi = (0.3, 0.3, 0.2, 0.1, 0.1)$, $\{\lambda_j\} = (10,2.5,1,0.5,0.2)$ and $\kappa = 10$. Pattern probability $\{\theta_{jw}\}$ is provided in Table \ref{sim3}.
	
	\tablehere{5}
	
	The model in this setting is also identifiable: A1 is satisfied since $\lambda_j$'s are different;
	A2 holds by noticing that there is no pattern with repeated actions;
	C1 and C2 are also satisfied automatically since the size of each equivalence class is 1.   
	
	\item
	In the fourth simulation setting, dictionary $\mathcal D$ includes patterns up to $3$-grams, with details provided in Table \ref{sim4}.
	We let $v_D = 50$, $L = 3$, $M_1 = 20, M_2 = 20, M_3 = 10$ and set $m = 1000$, $J = 5$. Other model parameters are set as follows, $\pi = (0.4, 0.3, 0.2, 0.05, 0.05)$, $\{\lambda_j\} = (1,1,1,1,1)$ and $\kappa = 10$. Pattern probability $\{\theta_{jw}\}$ is provided in Table \ref{sim4}.
	
	\tablehere{6}
	
	The model parameter in this setting is not identifiable: examinnes from Class 1 and Class 2 have almost the same pattern probabilities except for two patterns. Specifically, Class 1 does not have pattern $[1~2]$ and Class 2 does not have pattern $[2~3]$.
	Therefore, it fails to meet Condition C1. Thus we do not expect that all five classes can be recovered in this setting.
\end{enumerate}

We generate 50 datasets for each setting. To provide a more concrete sense of data, some descriptive statistics are given. The means of sentence length are 6.71, 6.89, 4.88 and 5.02 for Settings 1, 2, 3 and 4 respectively.
The maximum lengths of whole event sequences are around 178, 182, 135, 133 for Settings 1, 2, 3 and 4, respectively. 
The detailed procedures for generating the data sets are presented below.
\begin{itemize}
	\item[] \textbf{Data Generation Scheme}
	\item[] \textbf{Input}: $\mathcal D$, $m$, $\pi$, $\{\theta_{jw}\}$, $\lambda_{j}$ and $\kappa$.
	\item[] For $i = 1, \ldots, m$ do
	\begin{enumerate}
		\item Sample $z_i$ from the multinomial distribution with parameter $\pi$.
		\item Sample $K_i$ from the Poisson distribution with parameter $\kappa$.
		\item For $k = 1, \ldots, K_i$, do
		\begin{itemize}
			\item For each $w \in \mathcal D$, sample a indicator variable $u_w$ from the Bernoulli distribution with parameter $\theta_{z_i w}$.
			\item Randomly shuffle patterns in set $\{w | u_w = 1\}$ and get an ordered pattern sequence.
			\item Concatenate all patterns in above sequence and get the event sentence $E_{ik}$.
			\item Compute the length of $E_{ik}$, i.e., $n_{ik}$.
			\item Generate $T_{ik}$ recursively, such that $T_{ik,u} - T_{ik,u-1} \sim \textrm{exp}(\lambda_{z_i})$ for $u = 1, \ldots, n_{ik}$, where $T_{i1,0} = 0$ and $T_{ik,0} = T_{ik-1,n_{ik}} (k \geq 2)$.
		\end{itemize}
	\end{enumerate}
	\item[] \textbf{Output}: List of event sentences $\{E_{ik}\}$ and list of time sentences $\{T_{ik}\}$.
\end{itemize}

We set threshold $\tau = 1 / \sqrt{m}$ for each setting.
The performance of proposed model is evaluated through the following criteria.
\begin{itemize}
	\item[] \textit{Correct recovery}: percent of correctly identified patterns out of all true patterns.  
	\item[] \textit{False recovery}: percent of incorrectly identified patterns out of all identified patterns.
	\item[] \textit{$l$-gram hitting}: percent of correctly identified $l$-grams out of all true $l$-grams.
	\item[] \textit{Class recovery}: percent of recovering true number of latent classes.
	\item[] \textit{RMSE}: Root mean squared error of model parameters.
\end{itemize}

\tablehere{7}

\tablehere{8}

From Table \ref{simtab}, we can see that the proposed method can recover dictionary and model parameters well. The fact that ``correct recovery" is close to 1 and ``false recovery" is close 0 provides the empirical evidence that $\mathcal D$ is identifiable. The 2-grams, 3-grams are accurately recovered in all three settings. 
Their hitting rates are all close to 1.
The ``4-gram hitting" is also high in Setting 3. 
The estimates of mixing proportion $\pi$, response speed $\{\lambda_j\}$ and pattern probability $\{\theta_{jw}\}$ are close to their true values with small RMSEs.
These results provide the supporting evidence on the identifiability of model parameters.

Furthermore, in Settings 1 and 2, we compare the differences between results by fitting the model with/without times. 
Note that the model in Setting 1 remains estimable if times are ignored.
However, from Table \ref{simtab}, we can see the increase in RMSEs of parameters and the decrease in ``class recovery" when the times were taken out.  
This is consistent with the information theoretical results presented in Section 3.
In Setting 2, the true number of latent classes is not identifiable when we ignore the times, 
since Classes 1 and 2 are merged together into a single class, similarly for Classes 3 and 4.
The true $\lambda_j$'s are no longer estimable for Classes 1-4.
In fact, the estimator converges to its average value $\frac{2}{1/2 + 1/4} \approx 0.38$.
These results show that the inclusion of event times can lead to more accurate estimation and better identifiability. 
From Table \ref{simtab:set4}, under Setting 4, we can see that the proposed algorithm tends to find four classes instead of five classes (i.e. 66 percent of time the algorithm returns a four-class model). The estimate $\hat \pi$ indicates that Class 1 and Class 2 are merged together since we cannot distinguish between them. On the other hand, the underlying dictionary can still be recovered with high accuracy. 
 
From these results, we can see that the proposed algorithm is expected to recover the true latent classes and pattern structures when examinees from different classes have different patterns with different respond speeds.

\section{Real Data Analysis \label{sec:6}}

In this section, we apply the proposed model to the ``Traffic" item from PISA 2012 as described in Section 2.
The data were preprocessed as follows.
We removed those examinees who did not answer all three questions of the ``Traffic" item or did not take any actions, leaving 10048 remaining examinees.
In the raw data, each event corresponding to the map is a 0-1 vector with 23 entries.
Note that two consecutive vectors only differed at one position. We took their difference and represented event as the index on which the two consecutive vectors differ.
We view highlighting and unhighlighting as two different knowledge status of the examinee. As such, a sentence was defined as a subsequence of events where the examinee either consecutively highlighted roads or consecutively unhighlighted roads. A new sentence starts once the examinee changed from highlighting (unhighlighting) to unhighlighting (highlighting), or clicked ``reset".
The corresponding time sentence was defined accordingly.
An example of such data transformation is shown in Table \ref{cleandata}.
In our case, the observed data is $e_{1:N} = \{10, 8, 9, \ldots\}$ and $t_{1:N} = \{27.7, 28.6, 29.4, \ldots \}$. The corresponding observed event and time sentence sequences are $\{(10, 8, 9, 20), \ldots, (9, 8, 10)\}$ and  $\{ (27.7, 28.6, 29.4, 30.5), \ldots, (46.0, 47.7, 48.7)\}$. 
On average, each individual had about 10.4 sentences and clicked around 28.4 roads.

To apply the proposed method, we set $\tau = 1/\sqrt{10048}$ and $L = 3$ by observing that the correct meeting point is at most three roads away from each place marked in red.
Six classes were identified with the LTDM and were labeled in descending order according to their sizes.
Table \ref{result} provides a summary, such as mixing proportion, answer correct rate, average number of sentences (actions), etc., for the six classes.
The size of estimated dictionary was 82 with $M_1 = 23$, $M_2 = 39$ and $M_3 = 20$; see Appendix D for details.

\tablehere{9}

The fitted model appeared to satisfy the required assumptions and conditions. 
Assumption A1 was satisfied since the size of each equivalence was one by noticing that $\lambda_j$'s were significantly different from each other.
Assumption A2 was satisfied by observing that no pattern had repeated events.
Conditions C1 and C2 also held since the size of each equivalence class was one.

From Table \ref{result}, we can see that there was a substantial variation among the estimated intensities $\{\hat \lambda_j \}$.
Classes 2 and 3 had the highest response speed and the highest correct rate.
The examinees from Class 6 responded much slower compared to other classes but still had a decent chance to get the correct answer. It shows that they made efforts to solve this CPS item.
Examinees from Classes 4 and 5 solved the item with moderate speed. They had the lowest chance to get the correct answer.
The results indicate that the time information may be useful in characterizing examinees.


\tablehere{10}

We next look at the most frequent patterns (2-grams, 3-grams) in Table \ref{pattern:explanation}.
Examinees in Classes 2 and 3 had a much higher chance to solve the item. They could successfully identify three paths connecting the correct meeting point ``Park" to three original places, ``Silver", ``Nobel" and ``Lincoln". Note that there are two paths connecting ``Silver" and ``Park" with less than 15 minutes, i.e. $[20~9]$ and $[21~14~22]$. Examinees in Class 3 identified the former path and Examinees in Class 2 identified the latter. Pattern $[20~9]$ is shorter, thus Class 3 was the most efficient class. 
Examinees in Class 4 barely used any most frequent 2-gram or 3-gram patterns. They did not appear to have a good strategy that led to a lower rate of solving the item successfully.  
Examinees from Class 5 behaved differently from other classes. 
They identified paths connecting ``Silver" to ``Lincoln" or ``Nobel" , i.e., $[10~8~20]$, $[3~12~21]$, $[20~8~10]$ and $[21~12~3]$. 
In other words, these examinees found the second correct meeting point ``Silver". But unfortunately, most of them failed to identify ``Park".
Hence, they have the lowest rate to answer the item correctly. 
Examinees from Class 1 or Class 6 tended to have patterns such as $[10~8]$, $[6~19]$, $[3~4]$. These patterns were partial paths from three original places to the correct meeting point ``Park". 
This explains that Examinees in these two classes had a moderate chance to answer the question correctly. 

%

To summarize, the proposed approach produces a useful model and classifies examinees into interpretable classes.  
The results suggest that an efficient examinee (i.e. fewer actions, higher usage of frequent patterns) was more likely to successfully complete the task.  
Since ``Traffic" item tests the ability of ``Exploring and understanding" and ``Planning and executing" \citep{pisa2014}, our results suggest that this item achieved what it was intended for.

\section{Discussion \label{sec:7}}

In this paper, we proposed a new statistical model, the latent theme dictionary model, to deal with the process data and developed the NB-LTDM algorithm.
The new approach allows us to extract co-occurrent patterns and to classify individuals automatically based on data without pre-specifying the dictionary and the number of classes.
In addition, we established the theoretical properties of the proposed method, including model identifiability and consistency of parameter estimation.
The simulation results confirmed the theoretical findings.
We also applied the new method to the 2012 PISA ``Traffic" item and obtained meaningful results.   

%

It is easy to incorporate domain knowledge into our approach. 
If certain patterns are selected by experts, we can simply add them to the dictionary. On the other hand, if some patterns are known to be impossible or meaningless, they can be excluded from the dictionary.
Because of its generality, the proposed model can be applied in other context such as text mining and speech pattern recognition, where different articles and speeches could be clustered based on their word patterns. It can also be applied in user behavioral studies in e-commerce, online social networking, etc., where users' frequent daily action patterns can be extracted and user preference database can thus be built.

There are limitations in the proposed method that need to be addressed in further studies. 
First, although LTDM focuses on finding the ordered event pattern structure within a sentence,  we do not have an automated general rule for splitting the original event sequence to a list of sentences. 
Rather, the current sentence splitting method is ad hoc, relying on expert knowledge.
Second, it is an exploratory method to discover the underlying dictionary of action patterns and latent classes of examinees.
However, the current algorithm does not have the theoretically guaranteed convergence, though it works well empirically.
Third, in the current setting, the response speed only depends on examinee's latent class membership which may not fully capture the heterogeneity among examinees. 
A possible approach is to introduce individualized random effects to accommodate such heterogeneity.



\bibliographystyle{rss}
\bibliography{ltdm}


\vspace{\fill}\pagebreak

\appendix
\renewcommand{\theequation}{A\arabic{equation}}
\setcounter{equation}{0}
\renewcommand{\thesection}{\Alph{subsection}}
\setcounter{section}{0}

\section*{Appendix A: Conditions C1 and C2}

We provide the exact statements of conditions C1 - C2 in this appendix.
\begin{enumerate}
	\item[C1.a] For each equivalence class $[j]$ with size larger than 1, there exists a partition $\{\mathcal I_{[j],1}, \mathcal I_{[j],2}, \mathcal I_{[j],3} \}$ of 1-grams such that for any $e_1 \in \mathcal I_{[j],1}$, $e_2 \in \mathcal I_{[j],2}$ and $e_3 \in \mathcal I_{[j],3}$, sentence $E=(e_{l}, e_k),  l \neq k \in \{1,2,3\}$ and sentence $E = (e_1, e_2, e_3)$ admit only one separation. Cardinalities of three sets satisfy $|\mathcal I_{[j],1}|$, $|\mathcal I_{[j],2}|$ and $|\mathcal I_{[j],3}| \geq |[j]|$. Here $|[j]|$ is the cardinality of equivalence class $[j]$.
	\item[C1.b] Define $T$-matrices $T_{[j],1}$, $T_{[j],2}$ and $T_{[j],3}$ such that $T_{[j],k} [l,j_1] = \frac{\theta_{j_1 l}}{1 - \theta_{j_1 l}}$ for $e_l \in \mathcal I_{[j],k}$, $j_1 \in [j]$, and $k = 1, 2 ~\textrm{or}~ 3$. Matrices $T_{[j],1}$, $T_{[j],2}$ and $T_{[j],3}$ have full column rank.
	\item[C2.a] For each equivalence class $[j]$ with size larger than 1 and for any $l$-gram $w = [e_1~e_2 \ldots ~ e_l]$ with $l \geq 2$, there exists $\mathcal D_{[j],w}$ (the subset of 1-grams) such that (1) for any $e \in \mathcal D_{[j],w}$, sentence $E = (e_1, \ldots e_l, e)$ does not admit other separations containing $(l+1)$-gram or $l$-gram other than $w$; (2) cardinality of $\mathcal D_{[j],w}$  is greater than or equal to $|[j]|$.
	\item[C2.b] Define matrix $T_{[j],w}$ such that $T_{[j],w} [e,j_1] = \frac{\theta_{j_1 e}}{1 - \theta_{j_1 e}}$ for $e \in \mathcal D_{[j],w}$ and $j_1 \in [j]$.
	Matrix $T_{[j],w}$ has full column rank.
\end{enumerate}

Conditions C1 - C2 pertain to the dictionary and parameter structures.
Specifically, 
Condition C1.a puts the restrictions on 1-grams such that not all combinations of 1-grams are considered as patterns, which ensures the pattern frequency can be identified.
It is very similar to the sufficient conditions in identifiability of diagnostic classification models \citep[DCMs,][]{xu2017identifiability, fang2019identifiability}, where they require all items can be divided into three non-overlapping item sets. Here 1-gram can be viewed as the counterpart of item in DCMs.
Condition C2.a requires that each $l$-gram is not overlapped with other patterns to some extent and thus can be identified.
Conditions C1.b and C2.b require that the examinees from different groups should have different pattern frequencies.

The $T$-matrices here share the similar ideas to those in \cite{liu2012data, liu2013theory}. We use the following example to illustrate this idea.
\begin{example}
	Consider a 2-class model with $\lambda_1 = \lambda_2$ and $\mathcal D = \{ [a], [b], [c], [d], [e], [f], [a~b], [c~d], [e~f]\}$. Pattern probability $\{\theta_{jw}\}$ is 
	\[
	\begin{blockarray}{cccccccccc}
	& [a] & [b] & [c] & [d] & [e] & [f] & [a~b] & [c~d] & [e~f] \\
	\begin{block}{c(ccccccccc)}
	Class 1 & 0.5 & 0.5 & 0.5 & 0.5 & 0.5 & 0.5 & 0.5 & 0.5 & 0.5\\
	Class 2 & 0.75 & 0.25 & 0.75 & 0.25 & 0.75 & 0.25 & 0.5 & 0.5 & 0.5 \\
	\end{block}
	\end{blockarray}
	\]
	We claim this setting is identifiable.
	
	Notice that Classes 1 and 2 are in the same equivalence class [1].
	We can construct $\mathcal I_{[1],1} = \{[a], [b]\}$, $\mathcal I_{[1],2} = \{[c],[d]\}$, and $\mathcal I_{[1],3} = \{[e],[f]\}$. It is easy to check that their $T$-matrices satisfy
	\begin{eqnarray*}
		T_{[1],1} = T_{[2],1} = T_{[3],1} = 
		\begin{pmatrix}
			1 & 3 \\
			1 & 1/3\\
		\end{pmatrix}.
	\end{eqnarray*}
	Hence Condition C1 is satisfied, since they all have full column rank.
	For $w = [a~b]$, we can set $\mathcal D_{[1],w} = \{c,d\}$ by checking that both sentences $(a,b,c)$ and $(a,b,d)$ have only one separation. 
	Its $T$-matrix is
	\begin{eqnarray*}
		T_{[1],w}= 
		\begin{pmatrix}
			1 & 3 \\
			1 & 1/3\\
		\end{pmatrix},
	\end{eqnarray*}
	which is also full-column rank. Similarly, we can check it for $[c~d]$ and $[e~f]$. Thus Condition C2 is also satisfied.
	Furthermore, Assumption A1 holds since both classes contain all sentences in $\mathcal O$.
	Lastly, Assumption A2 obviously holds.  
\end{example}

\section*{Appendix B: Proofs}

To prove main theoretical results, we start with two lemmas which play key roles for dictionary and parameter identifiability. The proof of Lemma \ref{keylem} is presented at the end of this section.

\begin{lemma}[\cite{kruskal1977three}]\label{Lem}
	
	Suppose $A,B,C,\bar{A},\bar{B},\bar{C}$ are six matrices with $R$ columns.
	There exist integers $I_0$, $J_0$, and $K_0$ such that $I_0+J_0+K_0 \geq 2R+2$.
	In addition, every $I_0$ columns of $A$ are linearly independent, every $J_0$ columns of $B$ are linearly independent, and every $K_0$ columns of $C$ are linearly independent.
	Define a triple product to be a three-way array $[A,B,C] = (d_{ijk})$ where $d_{ijk}=\sum_{r=1}^{R} a_{ir} b_{ir} c_{kr}$.
	Suppose that the following two triple products are equal $[A,B,C]=[\bar{A},\bar{B},\bar{C}]$. Then, there exists a column permutation matrix $P$ such that $\bar{A}=AP\Lambda, \bar{B}=BPM, \bar{C}=CPN$, where $\Lambda, M, N$ are diagonal matrices and $\Lambda MN =$ identity.
	Column permutation matrix is right-multiplied to a given matrix to permute the columns of that matrix.
	
\end{lemma}

\begin{lemma}\label{keylem}
	Under Assumptions A1 and A2, it holds that $\mathcal O_{[j]} = \mathcal O_{[j]}^{'}$ if and only if $\mathcal D_{[j]} = \mathcal D_{[j]}$.  
\end{lemma}
Here we recall that $\mathcal O_{[j]}$ is the observed sentence set generated from equivalence class $[j]$ and $\mathcal D_{[j]}$ is the dictionary consisting of patterns from equivalence class $[j]$.

\begin{proof}[Proof of Theorem 1]
	For every model $\mathcal P = (\mathcal D, \{\theta_{jw}\}, \{\lambda_j\}, \pi, \kappa) \in \mathfrak P^0$, 
	we need to show that if there exists another model $\mathcal P^{'}$ such that 
	\begin{eqnarray}\label{main:obj}
	P(K|\kappa) \cdot \bigg\{ \sum_z  \pi_z \prod_{k=1}^K \bigg\{ \sum_{S_{k} \in \mathcal S_{k}} P(S_{k}, \tilde T_{k}| z) \bigg\} \bigg\} 
	=
	P(K|\kappa^{'}) \cdot \bigg\{ \sum_z  \pi_z^{'} \prod_{k=1}^K \bigg\{ \sum_{S_{k} \in \mathcal S_{k}} P(S_{k}, \tilde T_{k}| z) \bigg\} \bigg\}, \nonumber\\
	\end{eqnarray}
	it must hold $\mathcal P = \mathcal P^{'}$.
	
	We prove it through the following steps. 
	(1) $\kappa$-identifiability: we show that the parameter $\kappa$ is identifiable.
	(2) $\lambda$-identifiability: we prove that $\lambda_{[j]} = \lambda_{[j]}^{'}$ for any equivalence class $[j]$.
	(3) Dictionary identifiability: we show that $\mathcal O = \mathcal O^{'}$ implies $\mathcal D = \mathcal D^{'}$. (4) $\{\theta\}, \pi$-identifiability: we show that $\{\theta_{jw}\} \overset{p}{=} \{\theta_{jw}^{'}\}$
	and $\pi \overset{p}{=} \pi^{'} $.

	For $\kappa$-identifiability, we can see that the marginal distribution of $e_{1:N}$ and $t_{1:N}$ is
	\begin{eqnarray}
	P(e_{1:N}, t_{1:N}) = P(K|\kappa) \cdot \bigg\{ \sum_z  \pi_z \prod_{k=1}^K \bigg\{ \sum_{S_{k} \in \mathcal S_{k}} P(S_{k}, \tilde T_{k}| z) \bigg\} \bigg\}.
	\end{eqnarray}
	By taking $K = 0$, we have that $P(e_{1:N}, t_{1:N}) = P(K = 0)$.
	Then it must hold that
	\begin{eqnarray*}
		e^{- \kappa} = e^{- \kappa^{'}}.
	\end{eqnarray*}
	This implies that $\kappa = \kappa^{'}$.

	For $\lambda$-identifiability, we consider take $K = 1$ and an event sentence $E = (e)$ and $\tilde T = (t)$. Then, \eqref{main:obj} becomes
	\begin{eqnarray}
	P(K = 1|\kappa) \bigg\{ \sum_{j=1}^J  \pi_j \theta_{j e} \lambda \exp\{ - \lambda t\}  \bigg\} 
	= 
	P(K = 1|\kappa^{'}) \bigg\{ \sum_{j=1}^J  \pi_j^{'} \theta_{j e}^{'} \lambda^{'} \exp\{ - \lambda^{'} t\}  \bigg\}. \nonumber \\
	\end{eqnarray}
	By $\kappa$-identifiability, we further have 
	\begin{eqnarray}\label{lambda:obj}
    \sum_{j=1}^J  \pi_j \theta_{j e} \lambda_j \exp\{ - \lambda_j t\} 
	= 
    \sum_{j=1}^J  \pi_j^{'} \theta_{j e}^{'} \lambda_j^{'} \exp\{ - \lambda_j^{'} t\}
	\end{eqnarray}
	after simplification. Let $t \rightarrow \infty$, we must have that $\lambda_{[j_0]} = \lambda_{[j_0]}^{'}$, where $[j_0]$ is the equivalence class with minimum lambda value.
	Hence, we also have $\sum_{j \in [j_0]} \pi_j \theta_{je} = \sum_{j \in [j_0]} \pi_j^{'} \theta_{je}^{'}$. Then \eqref{lambda:obj} becomes 
	\begin{eqnarray*}
	\sum_{j \notin [j_0]}  \pi_j \theta_{j e} \lambda_j \exp\{ - \lambda_j t\} 
	= 
	\sum_{j \notin [j_0]}  \pi_j^{'} \theta_{j e}^{'} \lambda_j^{'} \exp\{ - \lambda_j^{'} t\}.
	\end{eqnarray*}
	By the similar strategy, we can show that $\lambda_{[j]} = \lambda_{[j]}^{'}$ for every equivalence class $[j]$. This gives $\lambda$-identifiability.
	
	For the dictionary identifiability, we would like to point out that its proof is not covered in \cite{deng2014association}. Therefore, we seek an alternative approach to prove it. 
	
	By taking $K = 1$, an arbitrary sentence $E \in \mathcal O$ and $\tilde T = (t, \ldots, t_{n_E})$ where $n_E$ is the sentence length. Then \eqref{main:obj} becomes
	\begin{eqnarray}\label{key:structure}
	\sum_{[j]}  [\sum_{j_1 \in [j]} \pi_{j_1} P(E|j)] (\lambda_{[j]})^{l_E} \exp\{ - \lambda_{[j]} n_E t\} 
	= 
	\sum_{[j]}  [\sum_{j_1 \in [j]} \pi_{j_1}^{'} P^{'}(E|j)] (\lambda_{[j]}^{'})^{n_E} \exp\{ - \lambda_{[j]}^{'} n_E t\} \nonumber. \\ 
	\end{eqnarray}
	Comparing the coefficients on both sides of \eqref{key:structure}, we then have
	\begin{eqnarray}\label{simple:obj}
	\sum_{j_1 \in [j]} \pi_{j_1} P(E|j) = \sum_{j_1 \in [j]} \pi_{j_1}^{'} P^{'}(E|j).
	\end{eqnarray} 
	This implies that $\mathcal O_{[j]} = \mathcal O_{[j]}^{'}$.
	By Lemma \ref{keylem}, we then have $\mathcal D_{[j]} = \mathcal D_{[j]}^{'}$. Notice that $\mathcal D = \cup_{[j]} \mathcal D_{[j]}$. It concludes the dictionary identifiability.
	
	For $\{\theta\}, \pi$-identifiability, we prove it by making use of \eqref{simple:obj}. In \eqref{simple:obj}, we take $E = (e)$ for $e \in \mathcal E$, $E = (e_1, e_2)$ with $e_1, e_2$ from different partition sets, and $E = (e_1, e_2, e_3)$ with $e_k \in \mathcal I_{[j],k}, (k = 1,2,3)$, sequentially.
	
	Without loss of generality, we suppose there is only one equivalence class. 
	According to Condition C1.a that $E$ only admits one separation, \eqref{simple:obj} can be simplified as
	\begin{eqnarray}\label{A1}
	\sum_j \eta_j \varphi_{je} = \sum_j \eta_j^{'} \varphi_{je}^{'}, & & ~\textrm{if $E = (e)$}  \\
	\sum_j \eta_j \varphi_{je_1} \varphi_{je_2} = \sum_j \eta_j^{'} \varphi_{je_1}^{'} \varphi_{je_2}^{'}, & & ~\textrm{if $E = (e_1, e_2)$} \\
	\sum_j \eta_j \varphi_{je_1} \varphi_{je_2} \varphi_{je_3} = \sum_j \eta_j^{'} \varphi_{je_1}^{'} \varphi_{je_2}^{'} \varphi_{je_3}^{'}, & & ~\textrm{if $E = (e_1, e_2, e_3)$}.
	\end{eqnarray}
	where we define $\eta_j = \pi_j \prod_e (1 - \theta_{je})$, $\varphi_{je} = \theta_{je}/(1 - \theta_{je})$.
	In addition, if we take $E$ to be an empty sentence, then it holds
	\begin{eqnarray}\label{A2}
	\sum_j \eta_j = \sum_j \eta_j^{'}.
	\end{eqnarray}
	It is not hard to write equations \eqref{A1} - \eqref{A2} in terms of tensor products of matrices, that is,
	\begin{eqnarray*}
		[\bar T_1, \bar T_2, \bar T_3] = [\bar T_1^{'}, \bar T_2^{'}, \bar T_3^{'}],
	\end{eqnarray*}
	where
	$$\bar T_1 = \left(
	\begin{array}{ccc}
	1 & \ldots & 1 \\
	\varphi_{1 v_1} & \ldots & \varphi_{J v_1}  \\
	\vdots & \vdots & \vdots \\
	\varphi_{1 v_{I_1}} & \ldots & \varphi_{J v_{I_1}}  \\
	\end{array}
	\right),
	$$
	$$\bar T_2 = \left(
	\begin{array}{ccc}
	1 & \ldots & 1 \\
	\varphi_{1 v_1} & \ldots & \varphi_{J v_1}  \\
	\vdots & \vdots & \vdots \\
	\varphi_{1 v_{I_2}} & \ldots & \varphi_{J v_{I_2}}  \\
	\end{array}
	\right),
	$$
	and
	$$\bar T_3 = \left(
	\begin{array}{ccc}
	1 & \ldots & 1 \\
	\varphi_{1 v_1} & \ldots & \varphi_{J v_1}  \\
	\vdots & \vdots & \vdots \\
	\varphi_{1 v_{I_3}} & \ldots & \varphi_{J v_{I_3}}  \\
	\end{array}
	\right)
	\cdot \Lambda,
	$$
	Here, $\Lambda$ is a $J$ by $J$ diagonal matrix with its $j$-th element equal to $\eta_j$.
	By Condition C1.b, column ranks of matrix $\bar T_1$, $\bar T_2$ and $\bar T_3$ are full column rank. 
	Therefore, by Lemma 1, we have that
	$$\bar T_1^{'} = \bar T_1 P A, ~ \bar T_2^{'} = \bar T_2 P B ~ \textrm{ and } ~ \bar T_3^{'} = \bar T_3 P C,$$
	where matrix $P$ is a column permutation matrix, $A,B$ and $C$ are diagonal matrices satisfying $ABC = I$.
	Since elements in first rows of $\bar T_1, \bar T_2, \bar T_1^{'}, \bar T_2^{'}$  are all ones, it implies $A = B = I$. Therefore, $C = I$ as well. Thus, we have $\bar T_1^{'} = \bar T_1 P$, $\bar T_2^{'} = \bar T_2 P$ and $\bar T_3^{'} = \bar T_3 P$. By comparing element-wisely, we can see that $\eta = \eta^{'}$ and $\{\varphi_{je}\} = \{\varphi_{je}^{'}\}$ up to a label switch. Further, $\{\theta_{je}\} \overset{p}{=} \{\theta_{je}^{'}\}$ due to the monotonicity relation between $\varphi_{je}$ and $\theta_{je}$.
	
	In the following, we prove that $\theta_{jw}$ is identifiable up to the same label switch for any pattern $w \in \mathcal D$ by induction. Suppose we have that $\theta_{jw}$ is generically identifiable when $w$ belongs to \{1-grams, ..., ($k$-1)-grams\}. We need to show that $\theta_{jw}$ is identifiable if $w$ is a $k$-gram.
	
	Let $\mathcal E_k$ be the sentence set including all $k$-grams in $\mathcal D$ and all possible combinations of $k$-gram and $1$-gram that are not in $\mathcal D$. 
	It is not hard to see that for each $E \in \mathcal E_k$, its separation can only be the combinations of all $m$-grams ($m < k$) or the combinations of $k$ gram and $1$-gram.
	\begin{eqnarray*}
		\sum_j \eta_j \varphi_{jw} = \sum_j \eta_j^{'} \varphi_{jw}^{'}, & & ~ \textrm{if $E = (w)$ and $w$ is $k$-gram;}  \\
		\sum_j \eta_j \varphi_{jv_1} \varphi_{jv_2} = \sum_j \eta_j^{'} \varphi_{jv_1}^{'} \varphi_{jv_2}^{'}, 
		& & ~ \textrm{if $E = (v_1, v_2)$, $v_1$ is a $k$-gram and $ v_2 \in \mathcal D_{v_1}$; }
	\end{eqnarray*}
	By previous results that $\eta \overset{p}{=} \eta^{'}$ and $\varphi_{v} \overset{p}{=} \varphi_{v}^{'}$ for those $v$'s are 1-grams, we could write above equations in the following matrix form, that is,
	\begin{eqnarray}\label{kgram}
	\bar T_w \tilde \varphi_{w} = \mathbf 0.
	\end{eqnarray}
	where $\tilde \varphi_{w} = (\varphi_{1w} - \varphi_{1w}^{'}, \ldots, \varphi_{Jw} - \varphi_{Jw}^{'})^T$ and
	$$ \bar T_w = \left(
	\begin{array}{ccc}
	\eta_1 & \ldots & \eta_J \\
	\varphi_{1v_1}\eta_1 & \ldots & \varphi_{Jv_1}\eta_J \\
	\vdots & \vdots & \vdots \\
	\varphi_{1v_J}\eta_1 & \ldots & \varphi_{Jv_J}\eta_J   \\
	\end{array}
	\right).
	$$
	Here, $v_1, \ldots, v_J$ are $J$ distinct 1-grams in $\mathcal D_v$.
	According to Condition C2.a and C2.b,
	\eqref{kgram} admits only one solution. Therefore, $\tilde \varphi_{w} = \mathbf 0$, which implies $ \theta_{w} \overset{p}{=} \theta_{w}^{'}$. Hence, we conclude that $\theta_{jw} = \theta_{jw}^{'}$ up to a label switch for all $w \in \mathcal D$.
	This concludes the $\{\theta\}, \pi$-identifiability.	
	By completing all steps, we establish the identifiability results.
\end{proof}

\bigskip

\begin{proof}[Proof of Theorem 2]
	We prove this result by two steps. In Step 1, we prove that 
	dictionary $\mathcal D$ can be estimated consistently.
	In Step 2, we show that the estimator, $(\{\hat \theta_{jw}\} , \{\hat \lambda_j\}, \hat \pi, \hat \kappa)$, is consistent.  
	Without loss of generality, we take compact set $\boldsymbol \Theta_c$ as $\boldsymbol \Theta_c$ = \{$\theta_{jw} \in [\eta, 1 - \eta], \pi_j \in [\eta, 1 - \eta], \sum_j \pi_j = 1, \lambda_j \in [c, C]$, $\kappa \in [c, C]$\}, where $\eta$, $c$, $C$ are some positive constants such that true model parameter is in $\Theta_c$. 
	
	\textbf{Proof of Step 1} We first introduce several useful event sets. 
	Define an event set $\Omega_{D}$,
	\begin{eqnarray}
	\Omega_{D} \equiv \{\omega | \mathcal O^{\ast} \subset \{E_{ik} | i = 1, \ldots, m, k = 1, \ldots, K_i \} \}.
	\end{eqnarray}
	In other words, all possible sentences are at least observed once on $\Omega_D$. 
	Define sets $\Omega_{E} = \{\omega | |\sum_i K_i| \geq  m \kappa / 2  \}$, $\Omega_K = \{\omega | K_i \leq K_0, i = 1, \ldots, m\}$, $\Omega_T = \{\omega | \tilde t_{ik,u} \leq t_0, \textrm{for all}~ i,k,u\}$ and $\Omega_b = \Omega_{E} \cap \Omega_K \cap \Omega_T \cap \Omega_D$. Next we show that $\Omega_b$ holds with high probability.
    Specifically, we show the upper bound for $P(\Omega_b^c)$ by decomposing $\Omega_b^c$ into four parts, i.e., 
    $$\Omega_b^c \subset \Omega_{E}^c \cup \Omega_{K}^c \cup 
    (\Omega_T^c \cap \Omega_K) \cup (\Omega_D^c \cap \Omega_E \cap \Omega_K \cap \Omega_T) .$$
	\begin{enumerate}
		\item 
		By large deviation property of Poisson random variable $\sum_i K_i$, we have that $P(\Omega_{E}^c) \leq \exp\{ - c m\}$ where $c = (1 - \log 2) \kappa^{\ast} / 2$.
		\item
		It is not hard to see that $P(\Omega_K^c) \leq m P(K_i > K_0) \leq m \exp\{- c K_0\}$ for some constant $c$ by using Poisson moment generating function.
		\item
		By union bound, we can get that 
		$P(\Omega_T^c \cap \Omega_K) \leq m K_0 l_{max} P(t_{ik,u} > t_0) \leq m K_0 l_{max} \exp\{- \lambda_{min} t_0\}$.
		Here $l_{max}$ is the longest sentence length and $\lambda_{min} = \min \{\lambda_j, j = 1, \ldots, J\}$.
		\item 
		We show that $P(E)$ has a positive lower bound for every sentence $E \in \mathcal O^{\ast}$. 
		Take an arbitrary $E$, we know that $P(E) =  \sum_{z = 1}^J \sum_{S \in \mathcal F(E)} P(S|z)$.
		Hence $P(E) \geq \prod_{w} (\theta_{jw}^{\ast})^{\mathbf 1\{w \in S\}} (1 - \theta_{jw}^{\ast})^{\mathbf 1\{w \notin S\}}$ for $S \in \mathcal F(E)$ and $j$ such that $E \in \mathcal O_j$.
		Thus $P(E)$ is bounded below by $\eta^{v_D}$, i.e., $P(E) \geq \eta^{v_D}$.
		
		Therefore, we have that $P(\Omega_D^c \cap \Omega_E \cap \Omega_K \cap \Omega_T) \leq |\mathcal O^{\ast}| P(E \notin \{E_{ik} | i = 1, \ldots, m, k = 1, \ldots, K_i \}) \leq |\mathcal O^{\ast}| (1 - \eta^{v_D})^{|\{E_{ik}\}|} \leq |\mathcal O^{\ast}| (1 - \eta^{v_D})^{m \kappa / 2} $.
	\end{enumerate}
	Hence, event $\Omega_b^c$ holds with probability at most 
	$2 \exp\{ - c m\} +  m \exp\{- c K_0\} + m K_0 l_{max} \exp\{- \lambda_{min} t_0\} + |\mathcal O^{\ast}| (1 - \eta^{v_D})^{m \kappa / 2}$.
	
	On event $\Omega_b$, we have that all sentence in $\mathcal O^{\ast}$ are at least observed once. By the dictionary identifiability from Theorem 1, we know that $\hat {\mathcal D} = \mathcal D^{\ast}$. In other words, $P(\hat {\mathcal D} \neq \mathcal D^{\ast}) \leq P(\Omega_b^c)$. This completes Step 1 by choosing $K_0 = (\log m)^2$ and $t_0 = (\log m)^2$.

	
	\textbf{Proof of Step 2} 
	For any fixed parameter $\Theta \equiv (\{\theta_{jw}\} , \{ \lambda_j\}, \pi, \kappa)$. Let $l(\Theta)$ denote the log-likelihood evaluated at $\Theta$.
	By identifiability we know that, $\mathbb E l(\Theta^{\ast}) > \mathbb E l(\Theta)$ for any distinct $ \Theta \in \boldsymbol \Theta_c$. By compactness of $B(\Theta^{\ast}, \delta)^c \cap \boldsymbol \Theta_c$ and continuity of $\mathbb E l(\Theta)$ (see \eqref{diff:expect}), there exists a positive number $\epsilon$ such that $\mathbb E l(\Theta) \leq \mathbb E l(\Theta^{\ast}) - 3 \epsilon$ for any $\Theta \in B(\Theta^{\ast}, \delta)^c \cap \boldsymbol \Theta_c$. In next, we prove the uniform convergence of $l(\Theta)$ to the expected value.

	By Bernstein inequality, we know that 
	\begin{eqnarray}
	P( \frac{1}{m} |\sum_i l_i(\Theta) - \mathbb E l_i(\Theta)| \geq  \sqrt{\mathrm{var}(l(\Theta))} \cdot x ) \leq 
	2 \exp\{- m x^2\}
	\end{eqnarray}
	holds point-wisely.	
	By compactness, $\textrm{var}(l(\Theta))$ is bounded by some constant $M$. Thus 
	\begin{eqnarray}\label{concentration}
	P( \frac{1}{m} |\sum_i l_i(\Theta) - \mathbb E l_i(\Theta)| \geq x ) \leq 
	2 \exp\{- m x^2 / M\}
	\end{eqnarray}
	for any fixed $\Theta$.
	
	Next, we consider bound the gap between $l_i(\Theta) - l_i(\Theta^{'})$. For notational simplicity, we omit subscript $i$ in the following displays. We know that 
	\begin{eqnarray}\label{diff:prob}
	& & l(\Theta) - l(\Theta^{'}) \nonumber \\
	& = & \log\{P(K)/P^{'}(K)\} + \log \{ \frac{\sum_j \pi_j \prod_{k = 1}^K P(E_k| j) P(T_k| j)}{\sum_j \pi_j^{'} \prod_{k = 1}^K P^{'}(E_k| j) P^{'}(T_k| j)} \} \nonumber \\
	& \leq & \log\{P(K)/P^{'}(K)\} + \log \{ \max_j \frac{\pi_j \prod_{k = 1}^K P(E_k| j) P(T_k| j)}{\pi_j^{'} \prod_{k = 1}^K P^{'}(E_k| j) P^{'}(T_k| j)}  \}   \nonumber \\
	& \leq & \log\{P(K)/P^{'}(K)\} + \max_j \log \{  \frac{\pi_j \prod_{k = 1}^K P(E_k| j) P(T_k| j)}{\pi_j^{'} \prod_{k = 1}^K P^{'}(E_k| j) P^{'}(T_k| j)}  \}  
	\end{eqnarray} 
	For $\|\Theta - \Theta^{'}\|_{\infty} \leq \delta_1$,
	we can see that $ \log\{P(K)/P^{'}(K)\} \leq C K \delta_1$ for some constant $C$.
	We can further show that $\log \{ P(E|j)/P^{'}(E|j) \} \leq C v_D \delta_1$ and $\log \{P(T|j)/P^{'}(T|j) \} \leq C l_{max} t_0 \delta_1$ on set $\Omega_b$. This is because 
	\begin{eqnarray}\label{diff:E}
	& & \log \{ P(E|j)/P^{'}(E|j) \} \nonumber \\
	& = & \log \frac{ \sum_{S \in \mathcal F(E)} P(S|j)}{\sum_{S \in \mathcal F(E)}P^{'}(S|j) } \nonumber \\
	& \leq & \max_{S \in \mathcal F(E)} \log\{ P(S|j) /  P^{'}(S|j) \} \nonumber \\
	& \leq & \max_{S \in \mathcal F(E)} \sum_{w} \max\{\log \{\theta_{jw} / \theta_{jw}^{'}\}, \sum_{w} \log \{(1 - \theta_{jw}) / (1 - \theta_{jw}^{'})\} \} \nonumber \\
	&\leq& C v_D \delta_1
	\end{eqnarray}
	and 
	\begin{eqnarray}\label{diff:T}
	& & \log \{ P(T|j)/P^{'}(T|j) \} \nonumber \\
	& = & \log \{ \prod_{u = 1}^{l_E} \lambda_j \exp\{ - \lambda_j t_u \} \} - \log \{ \prod_{u = 1}^{l_E} \lambda_j^{'} \exp\{ - \lambda_j^{'} t_u \} \} \nonumber \\
	&\leq& l_{max} (t_0 + 1) \delta_1.
	\end{eqnarray}
	With \eqref{diff:E} and \eqref{diff:T}, \eqref{diff:prob} becomes 
	\begin{eqnarray}\label{diff:together}
	l(\Theta) - l(\Theta^{'}) \leq \sum_k \{Cv_D \delta + l_{max} (t_0 + 1) \delta \} \leq C t_0 K_0 \delta_1,
	\end{eqnarray}
	by adjusting the constant.
	
	Next we prove that $\mathbb E l(\Theta)$ is a continuous function of $\Theta$. Define set $A_{k,t} = \{\omega | t - 1 \leq \max\{\tilde t_{k_1,u}; k_1 = 1, \ldots, k, u = 1, \ldots, n_{k_1}\} \leq t\}$ for $k,t = 1, 2, \ldots$. By algebraic calculation, for any $\Theta, \Theta^{'}$ with $\|\Theta - \Theta^{'}\|_{\infty} \leq \delta_1$, we have
	\begin{eqnarray}\label{diff:expect}
	& & \mathbb E l(\Theta) - \mathbb E l(\Theta^{'}) \nonumber \\
	& = & \sum_{k=0}^{\infty} P^{\ast}(K = k) \{ \int \sum_{E \in \mathcal O^{\ast}} P^{\ast}(E, T | k) \log\{\frac{P(E,T,k)}{P^{'}(E,T,k)} \} dT \} \nonumber \\
	& \leq & \sum_{k=0}^{\infty} P^{\ast}(K = k) \sum_{t = 1}^{\infty} \{ \int_{A_{k,t}} \sum_{E \in \mathcal O^{\ast}} P^{\ast}(E, T | k) \log\{\frac{P(E,T,k)}{P^{'}(E,T,k)} \} dT \} \nonumber \\
	& \overset{\eqref{diff:together}}{\leq} &
	\sum_{k=0}^{\infty} P^{\ast}(K = k) \sum_{t = 1}^{\infty} \{ \int_{A_{k,t}} \sum_{E \in \mathcal O^{\ast}} P^{\ast}(E, T | k) (C t k \delta_1) \} dT \} \nonumber \\
	& \leq & \sum_{k=0}^{\infty} P^{\ast}(K = k) \sum_{t = 1}^{\infty} (C t k \delta_1)  P^{\ast}(A_{k,t}) \nonumber \\
	&\leq& \sum_{k = 0}^{\infty} C k \delta_1 P^{\ast}(K = k) \sum_{t=1}^{\infty} k l_{max} \exp\{- \lambda_{min} t\} \nonumber \\
	&\leq& \sum_{k = 0}^{\infty} C \delta_1 l_{max} 1 / (1 - \exp\{-\lambda_{min}\}) k^2 P^{\ast}(K = k) \nonumber \\
	&\leq& C \delta_1 
	\end{eqnarray}
	by adjusting the constant.
	
	Thus we have $|\mathbb E l(\Theta) - \mathbb E l(\Theta^{'})| \leq \frac{\epsilon}{4}$ for any $\Theta, \Theta^{'}$ such that $\|\Theta - \Theta^{'}\| \leq \delta_2 \equiv\epsilon / (4C)$. Together with \eqref{diff:together}, we then have 
	\begin{eqnarray}\label{diff:all}
	& & \frac{1}{m} |\sum_i l_i(\Theta) - \mathbb E l_i(\Theta)| 
	- \frac{1}{m} |\sum_i l_i(\Theta^{'}) - \mathbb E l_i(\Theta^{'})| \nonumber \\
	&\leq& \frac{1}{m} |\sum_i l_i(\Theta) - \mathbb E l_i(\Theta) - (\sum_i l_i(\Theta^{'}) - \mathbb E l_i(\Theta^{'}))| \nonumber \\
	&\leq& \frac{1}{m} \sum_i \{|l_i(\Theta) - l_i(\Theta^{'})| + |\mathbb E l_i(\Theta) - \mathbb E l_i(\Theta^{'})| \} \nonumber \\
	&\leq& \epsilon/4 + \epsilon/4 \leq \epsilon/2,
	\end{eqnarray}
	when $\| \Theta - \Theta^{'}\|_{\infty} \leq \delta_3$.
	Here we take $\delta_3$ be $\min\{\epsilon/(4C t_0 K_0), \delta_2\}$.
	
	By the covering number technique, there exists a finite set $\mathcal N$ such that the distance of any two points from $\mathcal N$ is at least $\delta_3$. 
	Thus by \eqref{concentration}, we have 
	\begin{eqnarray*}
	P( \sup_{\Theta \in \mathcal N} \frac{1}{m} |\sum_i l_i(\Theta) - \mathbb E l_i(\Theta)| \geq \epsilon/2 ) \leq 
	2 |\mathcal N| \exp\{- m \epsilon^2 / (4 M)\}.
	\end{eqnarray*}
	Define the set $\Omega_g = \{\omega | \sup_{\Theta \in \boldsymbol \Theta_c } \frac{1}{m} |\sum_i l_i(\Theta) - \mathbb E l_i(\Theta)| \leq \epsilon \}$.
    Combined with \eqref{diff:all}, it further gives us that 
    \begin{eqnarray}\label{sup:concentration}
    P( \Omega_g^c ~\textrm{and}~
    \Omega_b) &\leq& 
    2 |\mathcal N|\exp\{- m \epsilon^2 / (4 M)\} \nonumber \\
    &\leq& 2 (\frac{D}{\delta_3})^{n_p} \exp\{- m \epsilon^2 / (4 M)\}
    \end{eqnarray}
	where $D$ is the diameter of $\boldsymbol{\Theta}_c$ and $n_p$ is the number of total model parameters.
	
	Lastly, by the definition of $\hat \Theta$ and \eqref{sup:concentration}, we have that 
	\begin{eqnarray}\label{obj:ineq}
	\frac{1}{m} \sum_i l_i(\hat \Theta) &\geq& \frac{1}{m} \sum_i l_i(\Theta^{\ast}) \nonumber \\
	&\geq& \frac{1}{m} \sum_i \mathbb E l_i(\Theta^{\ast}) - \epsilon \nonumber \\ &\geq& \sup_{\Theta \in \boldsymbol{\Theta}_c \cap B(\Theta^{\ast}, \delta)^c} \frac{1}{m} \sum_i \mathbb E l_i(\Theta) + 2 \epsilon \nonumber \\
	&\geq& \sup_{\Theta \in \boldsymbol{\Theta}_c \cap B(\Theta^{\ast}, \delta)^c} \frac{1}{m} \sum_i l_i(\Theta) + \epsilon
	\end{eqnarray}
	on $\Omega_b \cap \Omega_g$. In other words, \eqref{obj:ineq} implies that 
	\begin{eqnarray}\label{final:prob}
	P(\hat \Theta \in \boldsymbol{\Theta}_c \cap B(\Theta^{\ast}, \delta)^c) &\leq& P((\Omega_b \cap \Omega_g)^c) \nonumber \\
	&\leq& 2 (\frac{D}{\delta_3})^{n_p} \exp\{- m \epsilon^2 / (4 M)\} \nonumber \\
	& & + 2 \exp\{ - c m\} +  m \exp\{- c K_0\} + m K_0 l_{max} \exp\{ - \lambda_{min} t_0\} \nonumber \\
	& & + |\mathcal O^{\ast}| (1 - \eta^{v_D})^{m \kappa/2}.
	\end{eqnarray}
	By choosing $K_0 = (\log m)^2$ and $t_0 = (\log m)^2 $, the left hand side of \eqref{final:prob} goes to zero as $m \rightarrow \infty$.
	This concludes the proof.
\end{proof}

\bigskip

\begin{proposition}\label{prop:decomp}
	Under LTDM setting, the probability mass function $P(e_{1:N}, t_{1:N}; \Theta)$ can be written in the multiplicative form of $G(e_{1:N}; \Theta) F(e_{1:N}, t_{1:N}; \Theta_1)$ if and only if $\lambda_j = \lambda, j = 1, \ldots, J$.
	Here $\Theta_1$ is the model parameter excluding $\{\theta_{jw}\}$, $G$ and $F$ are some functions.
\end{proposition}
\begin{proof}[Proof of Proposition \ref{prop:decomp}]
	First, we write out the likelihood function 
	\begin{eqnarray}\label{eqn:prop1}
	& &P(e_{1:N}, t_{1:N}; \Theta) \nonumber\\
	&=& \frac{\kappa^{K}\exp\{-\kappa\}}{K!} \sum_{j=1}^J \pi_j \prod_{k=1}^K \{P(E_k; \{\theta_{jw}\}) P(\tilde T_k; \lambda_j) \} \nonumber \\
	&=& C(K, \kappa) \sum_{j=1}^J \pi_j \prod_{k=1}^K \{P(E_k; \{\theta_{jw}\}) P(\tilde T_k; \lambda_j) \},
	\end{eqnarray}
	where $C(K, \kappa)$ is some quantity depending on $K$ and $\kappa$.
	
	We first prove the sufficient part. Suppose $\lambda_j = \lambda, j = 1, \ldots, J$. Then \eqref{eqn:prop1} can be written as 
	\begin{eqnarray}
	& &P(e_{1:N}, t_{1:N}; \Theta) \nonumber\\
	&=& C(K, \kappa) \sum_{j=1}^J \pi_j \prod_{k=1}^K \{P(E_k; \{\theta_{jw}\}) P(\tilde T_k; \lambda_j) \} \\
	&=& C(K, \kappa) \{\sum_{j=1}^J \pi_j \prod_{k=1}^K \{P(E_k; \{\theta_{jw}\})\} \prod_{k=1}^K P(\tilde T_k; \lambda),
	\end{eqnarray}
	Hence we can take $G(e_{1:N; \Theta}) = C(K, \kappa) \{\sum_{j=1}^J \pi_j \prod_{k=1}^K \{P(E_k; \{\theta_{jw}\})\}$ and $F(e_{1:N}, t_{1:N}; \Theta_1) = \prod_{k=1}^K P(\tilde T_k; \lambda)$. This concludes the sufficient part.
	
	We next prove the necessary part. Suppose it is not true. In other words, 
	we can write $P(e_{1:N},t_{1:N}; \Theta) = G(e_{1:N}; \Theta) F(e_{1:N},t_{1:N}; \Theta_1)$ when $\lambda_j$'s are not all the same.  
	Without loss of generality, we assume $\lambda_1 < \lambda_2 \leq ... \leq \lambda_J$.
	By assumption, 
	\begin{eqnarray}\label{eqn:compare}
	C(K, \kappa) \sum_{j=1}^J \pi_j \prod_{k=1}^K \{P(E_k; \{\theta_{jw}\}) P(\tilde T_k; \lambda_j) \}
	= G(e_{1:N}; \Theta) F(e_{1:N},t_{1:N}; \Theta_1).
	\end{eqnarray}
	Divided by $\prod_{k=1}^K P(\tilde T_k; \lambda_1)$ on both sides of \eqref{eqn:compare}, we get
	\begin{eqnarray}\label{eqn:compare2}
	C(K, \kappa) \sum_{j=1}^J \pi_j \prod_{k=1}^K \{P(E_k; \{\theta_{jw}\}) \frac{P(\tilde T_k; \lambda_j)}{P(\tilde T_k; \lambda_1)} \}
	= G(e_{1:N}; \Theta) \frac{F(e_{1:N},t_{1:N}; \Theta_1)}{\prod_{k=1}^K P(\tilde T_k; \lambda_1)}
	\end{eqnarray}
	By letting $\tilde t_n \rightarrow \infty; n = 1, \ldots N$, the left hand side of \eqref{eqn:compare2} becomes
	$C(K, \kappa) \pi_1 \prod_{k=1}^K \{P(E_k; \{\theta_{1w}\})$.  Then we know that $G(e_{1:N}; \Theta)$ has the form of $C_1(e_{1:N};\Theta_1) \prod_{k=1}^K \{P(E_k; \{\theta_{1w}\})$. Plug this back to \eqref{eqn:compare}, we then get
	\begin{eqnarray*}
		C(K, \kappa) \sum_{j=1}^J \pi_j \prod_{k=1}^K \{P(E_k; \{\theta_{jw}\}) P(\tilde T_k; \lambda_j) \}
		= C_1(e_{1:N};\Theta_1) \prod_{k=1}^K \{P(E_k; \{\theta_{1w}\}) F(e_{1:N},t_{1:N}; \Theta_1).
	\end{eqnarray*}
	Notice that the left hand side of above equation is a polynomial of $\{\theta_{jw}\}$'s and right hand side is a polynomial of $\{\theta_{1w}\}$'s.
	Hence it must hold that  
	$\pi_j \prod_{k=1}^K P(\tilde T_k, \lambda_j) \equiv 0$ for $j = 2, \ldots, J$, which is impossible when $\pi_j > 0$. Thus it contradicts with the assumption.
	This concludes the proof of necessary part.
\end{proof}

\bigskip

\begin{proof}[Proof of Lemma 2]
	For any pattern $w$ in $\mathcal D_{[j]}$, we need to show it also belongs to $\mathcal D_{[j]}^{'}$. 
	It is easy to see that if $w$ has the form of $[A]$, then it must belong to $\tilde D$ since $[A]$ only admits one separation.
	In the following, we only need to consider $w = [e_1 e_2 \ldots e_{l_w}]$ such that $e_1, \ldots ,e_{l_w}$ ($l_w \geq 2$) are different according to Assumption A2.
	With out loss of generality, we assume $w$ belongs to Class $j$.
	
	Let $\check O_j$ denote the longest sentence generated by $\mathcal D_j$ satisfying that (1) each event belongs to $\mathcal E - \{e_1\}$; (2) the length of $\check O_j$ is at least $n_{j,e_1}$ (See $n_{j,e_1}$'s definition in Assumption A1).
	Notice that $\check O_j$ may not be unique, we only need to consider one of them.
	Let $\grave{O}_j$ be the sentence such that it has form $(Q_1 Q_2)$ such that (1) $Q_1$ contains $\check O_j$ as its subsentence; (2) Each event in $Q_1$ belongs to $\mathcal E - \{e_1\}$; (3) $Q_1$ is longest possible; (4) The first event of $Q_2$ is $e_1$ (Be empty if it does not exist.); (5) $Q_2$ is shortest possible.
	Notice $\grave O_j$ may not be unique, we only need to consider one of them.
	Next we consider the decomposition of sentence $O_j = (\grave O_j w)$. By aid of $O_j$, we can show that $w$ must belong to $\mathcal D_{[j]}^{'}$.
	
	Since $\mathcal O_{[j]} = \mathcal O_{[j]}^{'}$, we know that $O_j \in \mathcal O_{[j]}^{'}$. Without loss of generality, we also assume that $O_j$ appears in $j$-th class of model $\mathcal P^{'}$. We claim that
	$O_j$ only has separations in form $\{\mathcal S(\grave O_j), w\}$ in $\mathcal D^{'}$. ($\mathcal S(O)$ is one realization of separation for $O$.) If not, then we must have the following cases.
	
	Case 1: There is a separation $S \in \mathcal F(O_j)$ such that $S = \{\mathcal S(R_1), w_1\}$ where $w_1$ is contained in $w$. By Assumption A2, we know that $w_1$ does not consist of $e_1$. Then, we consider sentence $(w_1 R_1)$. It is in $\mathcal O_{[j]}^{'}$, then it must in $\mathcal O_{[j]}$. By Assumption A1, we know that $(w_1 R_1)$ must belong to $\mathcal O_j$, since it contains $\check O_j$. Then, $(w_1 R_1)$ can be written in form of $(Q_1 Q_2)$ with longer $Q_1$. This contradicts with the definition of $\grave O_j$. Therefore, Case 1 cannot happen.
	
	Case 2: There is a separation $S \in \mathcal F(O_j)$ such that $S = \{\mathcal S(R_2), w_2\}$ where $w_2$ contains $w$. We further consider the following four situations.
	\begin{itemize}
		\item[2.a]
		 Suppose $R_2$ contains $\check O_j$ as its sub-sentence and does not contain events $u_1$. Since $\mathcal O_{[j]} = \tilde{\mathcal O}_{[j]}$, we know that $R_2$ must belong to $\mathcal O_{[j]}$. By Assumption A1, $R_2$ is also in $\mathcal O_j$. Then, it leads to contradiction since $R_2$ is longer than $\check O_j$. 
		\item[2.b]
		Suppose $R_2$ contains $\check O_j$ as its sub-sentence and contains events $u_1$. $R_2$ is also in $\mathcal O_j$ for the same reason as before. This time, $R_2$ can be written in the form of $(Q_1 Q_2)$ with shorter $Q_2$, which contradicts with the definition of $\grave O_j$.
		\item[2.c] 
		Suppose $R_2$ is contained in $\check O_j$. If $R_2$ is the longest sentence generated by $\mathcal D_{j}^{'}$ without $e_1$, then by Assumption A1 we know $\check O_j$ must also belong to this class. Therefore, $R_2$ is not the longest sentence. It implies that there exists a pattern $\tilde w$ with events in $\mathcal E_w$ in $\mathcal D^{'}_j$ does not contribute to $R_2$. Therefore, sentence $(\tilde w R_2 w_2)$ containing $\check O_j$ must belong to $\mathcal O_j$. By using the same argument, we know that it can also be written in the form $(Q_1 Q_2)$ with longer $Q_1$. This contradicts with the definition of $\grave O_j$.
		\item[2.d] Suppose $R_2 = \check O_j$. If $Q_2$ is not empty, then $w_2$ contains two $e_1$'s. It contradicts with Assumption A2. If $Q_2$ is empty, then $w_2 = w$ exactly.
	\end{itemize}
	Hence, we conclude the proof of this lemma.
\end{proof}

\section*{Appendix C: Parameter Estimation in NB-LTDM}

In the inner part of the NB-LTDM Algorithm, we adopt a nonparametric Bayes method
which is used to avoid selection of a single finite number of mixtures $J$. Therefore, we replace finite mixture components by an infinite mixture, that is,
\begin{eqnarray*}
	P(e_{1:N}, t_{1:N}) &=& P(K | \kappa) \sum^\infty_{j=1} v_j \prod_{k=1}^K P(E_k, T_k| j), \quad i=1,\ldots,m, \\
	\sum_{j=1}^{\infty} v_j & = & 1.
\end{eqnarray*}

For the choice of prior, we specify $\theta_{jw} \sim \text{Unif}(0,1)$, $\kappa \sim \mathrm{Ga}(1,1)$, $\lambda_j \sim \mathrm{Ga}(1,1)$ and $v = (v_1, \ldots) \sim Q$ where $Q$ corresponds to a Dirichlet process. The stick-breaking representation, introduced by \cite{sethuraman1994constructive}, implies that $v_j = V_j \prod_{l < j} (1- V_l)$ with $V_j \sim \text{Beta}(1,\alpha)$ independently for $j=1,\ldots,\infty$, where $\alpha >0$ is a precision parameter characterizing $Q$.

Hence, our nonparametric Bayesian latent theme dictionary model can be written in the following hierarchical form:
\begin{eqnarray*}
	{ S_{ik},T_{ik}| z_i = j, \{\theta_{jw}\}, \{\lambda_j\} }& \sim &\frac{1}{n_{S_{ik}}!} \prod_{w=1}^{v_D} [\theta_{jw}^{\mathbf 1_{\{w \in S_{ik}\}}} (1 - \theta_{jw})^{\mathbf 1_{\{w \not\in S_{ik}\}}}] \cdot \prod_{u = 1}^{n_{ik}} [\lambda_{j}e^{-\lambda_{j} \tilde t_{ik,u}}] \\
	E_{ik},T_{ik}| z_i, \{\theta_{jw}\}, \{\lambda_j\} & \sim & P(E_{ik}, T_{ik}| z_i) = \sum_{S \in \mathcal{F}(E_{ik})} P(S_{ik}, T_{ik}|z_i, \{\theta\}, \{\lambda\}), \quad i \in [m]; k \in [K_i]\\
	z_i &\stackrel{iid}{\sim} & \sum^\infty_{j=1} V_j \prod_{l < j} (1 - V_l) \delta_j, \quad i=1,\ldots,m \\
	V_j & \stackrel{iid}{\sim} & \text{Beta}(1,\alpha), \quad j=1,\ldots,\infty\\
	\theta_{hw} & \stackrel{iid}{\sim} & \text{Unif}(0,1), \quad j=1,\ldots,\infty, w=1,\ldots,v_D \\
	{ \lambda_{j} } &\stackrel{iid}{\sim}& \textrm{Ga}(1,1), \quad j = 1, \ldots, \infty \\
	\alpha & \sim & \text{Ga}(1,1) \\
	\kappa & \sim & \text{Ga}(1,1).
\end{eqnarray*}
where $\delta_j(\cdot)$ is the Dirac measure at $j$.

We use a data augmentation Gibbs sampler \citep{walker2007sampling} to update all parameters and latent variables. Specifically, we introduce a vector of latent variables $u=(u_1, \ldots, u_N)$, where $u_i \stackrel{iid}{\sim} U[0,1]$.  The full likelihood becomes
\begin{eqnarray*}
	& &\prod_{i=1}^m \bigg \{ \big\{ \prod_{k = 1}^{K_i} \mathbf 1_{\{ u_i < v_{z_i}\}} P(S_{ik},T_{ik}|z_i, \{\theta_{jw}\}, \{\lambda_j\}) \mathbf 1_{\{\mathcal F(E_{ik}) \in \mathcal S_{ik}\}} \big\}
	\frac{\kappa^{K_i} \exp\{-\kappa\}}{K_i !} \bigg \} \\
	& & \cdot  
	\prod_{j} \big\{ f_{Be}(V_j; \alpha) f_{Ga}(\lambda_j) \prod_{w = 1}^{v_D} f_u(\theta_{jw}) \big\} f_{Ga}(\kappa) f_{Ga}(\alpha),
\end{eqnarray*}
where $f_{Ga}$ is the density of $\mathrm{Ga}(1,1)$, $f_u$ is the density of $\mathrm{Unif}(0,1)$ and $f_{Be}(\cdot;\alpha)$ is the density of $\mathrm{Beta}(1, \alpha)$.

Then Gibbs sampler iterates through the following steps:
\begin{enumerate}
	\item Update $u_i$, for $i=1,\ldots,m$, by sampling from $U(0,v_{z_i})$.
	\item Update $\theta_{hw}$, for $h=1,\ldots,j^*, w=1\ldots, v_D$, by sampling from
	\begin{equation*}
	\text{Beta} \bigg( \sum^n_{i:z_i = h} \sum^{K_i}_{k=1} 1(\theta_{hw} \in S_{ik}) + 1, \sum^n_{i:z_i = h} \sum^{K_i}_{k=1} 1(\theta_{hw} \notin S_{ik} ) + 1 \bigg).
	\end{equation*}
	\item { Update $\lambda_j$, for $j = 1, \ldots, j^{\ast}$ ($j^{\ast} = \max\{z_i\}$), by sampling from 
		\begin{eqnarray*}
			\textrm{Ga}(1 + \sum_{i:z_i = j} \sum_{k=1}^{K_i} l_{T_{ik}}, 1 + \sum_{i:z_i = j} \sum_{k=1}^{K_i} \sum_{u = 1}^{l_{T_{ik}}} T_{ik,u}).	
		\end{eqnarray*}
	}
	\item Update $V_j$, for $j=1,\ldots,j^*$, by sampling from $\text{Beta}(1,\alpha)$ truncated to fall into the interval
	\begin{equation*}
	\bigg[\max_{i:z_i = j} \frac{u_i}{\prod_{l < j} (1-V_l)},
	1- \max_{i: z_i > j} \frac{u_i}{V_{z_i} \prod_{l < z_i, l \neq j}(1-V_l)} \bigg].
	\end{equation*}
	\item Update $z_i$, for $i=1,\ldots,m$, by sampling from
	\begin{equation*}
	P(z_i = j|e_{1:N_i},t_{1:N_i},\mathbf S_i, \{\theta_{jw}\},V,u,z_{-i}) = \frac{1(j \in A_i) \prod^{K_i}_{k=1} P(S_{ik},T_{ik}|\{\theta_{lw}\}, \lambda_j) }{\sum_{l \in A_i} \prod^{K_i}_{k=1} P(S_{ik}, T_{ik}|\{\theta_{lw}\}, \lambda_j) }1(S_{ik} \in \mathcal{F}(E_{ik})),
	\end{equation*}
	where $A_i := \{ j: v_j > u_i\}$.
	To identify the elements in $A_1,\ldots,A_m$, first update $V_j$ for $j=1,\ldots,\tilde{k}$, where $\tilde{j}$ is the smallest value satisfying
	\begin{equation}\label{find_A}
	\sum^{\tilde{j}}_{j=1} v_j > 1 - \min \{u_1,\ldots, u_m \}.
	\end{equation}
	
	Therefore, $1,\ldots,\tilde{j}$ are the possible values for $z_i$. Note that we have already updated $V_j$ for $j=1,\ldots,j^*$. Therefore, we first check if $j^*$ satisfies (\ref{find_A}). If yes, then we do not have to sample more; otherwise sample $V_j \sim \text{Beta}(1,\alpha)$ for $j=j^*+1,\ldots$ until (\ref{find_A}) is satisfied. 
	In this case, we also have to sample $\theta_{j w}$ from $U(0,1)$ and {$\lambda_j$ from $\textrm{Ga}(1,1)$} for $j = j^*+1,\ldots, \tilde{j}$ and $w =1,\ldots,v_D$ in order to compute $P(S_{ik}, T_{ik} | \theta_j, \lambda_j)$ for $j=j^* +1,\ldots,\tilde{j}$.

	\item Update $S_{ik}$, for $i=1,\ldots,m$, $k=1,\ldots,K_i$, by sampling from
	\begin{equation*}
	{ P(S_{ik} = S|E_{ik}, \theta_{z_i} ) } =  \frac{P(S_{ik} = S |\theta_{z_i})}{ \sum_{S' \in \mathcal{F}(E_{ik})} P(S_{ik} = S' | \theta_{z_i}) }  1(S \in \mathcal{F}(E_{ik})).
	\end{equation*}
	\item Update $\kappa$, which follows gamma distribution $\textrm{Ga}(1 + \sum_i K_i, 1 + m)$.
	\item Sample $\alpha$ from posterior $\textrm{Ga}(1 + j^{\ast}, 1  - \sum_{j=1}^{j^{\ast}} \log(1 - V_{j}))$.
\end{enumerate}

\section*{Appendix D: Estimated Dictionary in Traffic Item}

In the ``Traffic" item, the NB-LTDM algorithm found a dictionary $\hat {\mathcal D}$ with $\hat v_D = 82$.
For each pattern $w$ in $\hat {\mathcal D}$, we classified the six classes into two clusters based on their pattern probabilities $\theta_{jw}(j = 1, \ldots, 6)$.
Those classes with high pattern probabilities are clustered together and shown in Table \ref{realdata:dictionary}.

\tablehere{11}

\vspace{\fill}\pagebreak


\section*{\label{pre:model} Appendix E: Latent Class Model and Theme Dictionary Model}

In this section,
we briefly recall two popular models, latent class model \citep[LCM;][]{gibson1959three} and theme dictionary model \citep[TDM;][]{deng2014association}, which are related with the proposed LTDM.

Widely adopted in biostatistics, psychometrics and machine learning literature \citep[e.g.,][]{goodman1974exploratory, vermunt2002latent, templin2010diagnostic}, 
LCM relates a set of observed variables to a discrete latent variable, which is often used for indicating the class label. LCM assumes a local independence structure, i.e.,
\begin{eqnarray*}
	P(X_1, \ldots, X_K | Z) = \prod_{k = 1}^K P(X_k | Z),
\end{eqnarray*}
where $X_1, \ldots, X_K$ are $K$ observed variables and $Z$ is a discrete latent variable with density $P(Z = j) = \pi_j, ~ j = 1, \ldots J$. Thus, the joint (marginal) distribution of $X_1, \ldots, X_K$ takes form
\begin{eqnarray*}
	P(X_1, \ldots, X_K) =  \sum_{j = 1}^J \bigg\{ \pi_j \prod_{k = 1}^K P(X_k | Z = j)\bigg \}.
\end{eqnarray*}

For TDM \citep{deng2014association}, it typically handles observations known as words/events. It can be used for identifying associated event patterns. 
The problem of finding event associations is also known as market basket analysis \citep{piatetsky1991discovery, hastie2005elements, chen2005market}. 
Under TDM, a \textit{pattern} is a combination of several events. A collection of distinct patterns forms a \textit{dictionary}, $\mathcal D$. An \textit{observation}, $E$, is a set of events.
In TDM, we observe $E$ but do not know which patterns it consists of. In other words, $E$ could be split into different possible partitions of patterns. For each possible partition, we call it a separation of $E$.
The collection of multiple observations $\mathbf E = \{E_1, \ldots, E_K\}$ forms a document.
TDM does not take into account event ordering. For example, $E = (A, B, C)$ is an observation with three events, $A$, $B$ and $C$. TDM treats $E^{'} = (C, B, A)$ as the same observation as $E$. Consequently, patterns are also unordered. For instance, patterns $[A~B]$ and $[B~A]$ are viewed as the same.
TDM postulates that a pattern appears in an observation at most once. Let $\theta_w \in [0,1]$ be the probability of pattern that appears in an observation.
The probability distribution of one separation $S$ for observation $E$ is defined to be
\begin{eqnarray}\label{eq:TDM}
P(S) = \prod_{w \in \mathcal D} \theta_w^{\mathbf 1_{ \{ w \in S \} }} (1 - \theta_w)^{\mathbf 1_{ \{ w \notin S \} }}.
\end{eqnarray}
Since separation $S$ is not observed, the marginal probability of $E$ is
\begin{eqnarray*}
	P(E) = \sum_{S \in \mathcal F(E)} P(E, S) = \sum_{S \in \mathcal F(E)} \prod_{w \in \mathcal D} \theta_w^{\mathbf 1_{ \{ w \in S \} }} (1 - \theta_w)^{\mathbf 1_{ \{ w \notin S \} }},
\end{eqnarray*}
where $\mathcal F(E)$ is the set of all possible separations for $E$.
Furthermore, observations are assumed to be independent, i.e., for $\mathbf E = \{E_1, \ldots, E_K\}$, 
\begin{eqnarray*}
	P(\mathbf E) = \prod_{k=1}^K P(E_k).
\end{eqnarray*}

\begin{table}
	\caption{\label{sim1} Simulation setting 1}
	\centering
	\begin{adjustbox}{max width=\textwidth}
		\begin{tabular}{c | c | c c c c c c }
			\hline\hline
			& &  1-10: 1-gram & 11-20: 1-gram &  21-30: 2-gram&  31-40: 2-gram &  41-45: 3-gram &  46-50: 3-gram \\
			\hline
			\multirow{5}{*}{$\theta_{jw}$} & 1 & 0.3 & 0 & 0.2 & 0 & 0 & 0 \\
			& 2 & 0 & 0.3 & 0 & 0.2 & 0 & 0 \\
			& 3 & 0.2 & 0.2 & 0.05 & 0.05 & 0.001 & 0.001  \\
			& 4 & 0.05 & 0.05 & 0 & 0 & 0.3 & 0  \\
			& 5 & 0 & 0 & 0.03 & 0.03 & 0 & 0.3 \\
			\hline
		\end{tabular}
	\end{adjustbox}
 
	\begin{adjustbox}{max width=\textwidth}
	\begin{tabular}{c | l | l }
		\hline\hline
        \multicolumn{3}{c}{Specification of dictionary} \\
		\hline
		\multirow{5}{*}{$\mathcal D$} & 1-grams (1-20) & [1], [2], \ldots, [20] \\
		& 2-grams (21 - 30) & [1~2], [2~3], [3~4], [4~5], [5~1], [6~7], [7~8], [8~9], [9~10], [10~6]  \\
		& 2-grams (31 - 40) & [11~12], [12~13], [13~14], [14~15], [15~11], [16~17], [17~18], [18~19], [19~20], [20~11] \\
		& 3-grams (41 - 45) & [11~12~14], [12~13~15], [13~14~12], [14~15~11], [15~11~13]  \\
		& 3-grams (46 - 50) & [1~2~4], [2~3~5], [3~4~7], [3~9~6], [2~5~6] \\
		\hline
	\end{tabular}
  \end{adjustbox}
\end{table}

\begin{table}
	\caption{\label{sim2} Simulation setting 2}
	\centering
	\begin{adjustbox}{max width=\textwidth}
		\begin{tabular}{c | c | c c c c c c }
			\hline\hline
			& &  1-10: 1-gram & 11-20: 1-gram &  21-30: 2-gram&  31-40: 2-gram &  41-45: 3-gram &  46-50: 3-gram \\
			\hline
			\multirow{5}{*}{$\theta_{jw}$} & 1 & 0.3 & 0 & 0.2 & 0 & 0 & 0 \\
			& 2 & 0.3 & 0 & 0.2 & 0 & 0 & 0 \\
			& 3 & 0 & 0.3 & 0 & 0.2 & 0 & 0 \\
			& 4 & 0 & 0.3 & 0 & 0.2 & 0 & 0 \\
			& 5 & 0.05 & 0.05 & 0 & 0 & 0.3 & 0  \\
			& 6 & 0 & 0 & 0.03 & 0.03 & 0 & 0.3 \\
			\hline
		\end{tabular}
	\end{adjustbox}
	\begin{adjustbox}{max width=\textwidth}
	\begin{tabular}{c | l | l }
		\hline\hline
		\multicolumn{3}{c}{Specification of dictionary} \\
		\hline
		\multirow{5}{*}{$\mathcal D$} & 1-grams (1-20) & [1], [2], \ldots, [20] \\
		& 2-grams (21 - 30) & [1~2], [2~3], [3~4], [4~5], [5~1], [6~7], [7~8], [8~9], [9~10], [10~6]  \\
		& 2-grams (31 - 40) & [11~12], [12~13], [13~14], [14~15], [15~11], [16~17], [17~18], [18~19], [19~20], [20~11] \\
		& 3-grams (41 - 50) & [1~2~4], [2~3~5], [3~4~7], [3~9~6], [2~5~6] \\
		& 3-grams (46 - 50) & [11~12~14], [12~13~15], [13~14~12], [14~15~11], [15~11~13]  \\
		\hline
	\end{tabular}
    \end{adjustbox}
\end{table}

\begin{table}
	\caption{\label{sim3} Simulation setting 3}
	\begin{adjustbox}{max width=\textwidth}
		\begin{tabular}{c | c | c c c c c c c c c}
			\hline\hline
			& &  1-15: 1-gm & 15-30: 1-gm &  31-35: 2-gm&  36-50: 2-gm &  50-60: 2-gm &  61-70: 3-gm & 71-75: 3-gm & 76-85: 4-gm & 86-90: 4-gm \\
			\hline
			\multirow{5}{*}{$\theta_{jw}$} & 1 & 0.15 & 0 & 0 & 0 & 0 & 0.06 & 0.06 & 0 & 0  \\
			& 2 & 0 & 0.15 & 0.06 & 0.06 & 0.06 & 0 &0 &0 &0  \\
			& 3 & 0.05 & 0.05 & 0.05 & 0.001 & 0.001 & 0.05 &0.001 &0.001 &0.001  \\
			& 4 & 0 & 0 & 0.03 & 0.03 & 0 & 0 &0 &0.05 &0 \\
			& 5 & 0.04 & 0.04 & 0 & 0 & 0 & 0 &0 &0 & 0.1 \\
			\hline
		\end{tabular}
	\end{adjustbox}
	\begin{adjustbox}{max width=\textwidth}
	\begin{tabular}{c | l | l }
		\hline\hline
		\multicolumn{3}{c}{Specification of dictionary} \\
		\hline
		\multirow{5}{*}{$\mathcal D$} & 1-grams (1-30) & [1], [2], \ldots, [30] \\
		& 2-grams (31 - 45) & [1~2], [2~1], [2~3], [3~2], [3~4], [4~3], [4~5], [5~4], [5~1], [1,5], [6~7], [7~8], [8~9], [9~10], [10~6]  \\
		& 2-grams (46 - 60) & [11~12], [12~13], [13~14], [14~15], [15~11], [16~17], [17~18], [18~19], [19~20], [20~11],
		[1~11], [2~12], [3~13], [4~14], [5~15] \\
		& 3-grams (61 - 68) & [1~2~4], [2~3~5], [3~4~7], [3~9~6], [2~5~6], [2~1~4], [3~2~5], [4~2~7] \\
		& 3-grams (69 - 75) & [11~12~14], [12~13~15], [13~14~12], [14~15~11], [15~11~13], [12~11~14], [13~12~15]  \\
		& 4-grams (76 - 83) & [1~2~3~4], [2~3~5~1], [3~4~7~1], [3~9~6~2], [2~5~6~4], [3~4~1~2], [5~1~7~8], [6~9~3~4] \\
		& 4-grams (84 - 90) & [11~12~13~14], [12~13~15~11], [13~14~17~11], [24~25~26~27], [24~26~28~30], [11~16~21~26], [16~11~26~21]  \\
		\hline
	\end{tabular}
\end{adjustbox}
\end{table}

\begin{table}
	\caption{\label{sim4} Simulation setting 4}
	\centering
	\begin{adjustbox}{max width=\textwidth}
		\begin{tabular}{c | c | c c c c c c }
			\hline\hline
			& &  1-10: 1-gram & 11-20: 1-gram &  21-30: 2-gram&  31-40: 2-gram &  41-45: 3-gram &  46-50: 3-gram \\
			\hline
			\multirow{5}{*}{$\theta_{jw}$} & 1 & 0.15 & 0.15 & 0.1(except 21) & 0 & 0 & 0 \\
			& 2 & 0.15 & 0.15 & 0.1(except 22) & 0 & 0 & 0 \\
			& 3 & 0.1 & 0.1 & 0 & 0.15 & 0.0 & 0.0  \\
			& 4 & 0.05 & 0.05 & 0 & 0 & 0.3 & 0  \\
			& 5 & 0 & 0 & 0.03 & 0.03 & 0 & 0.3 \\
			\hline
		\end{tabular}
	\end{adjustbox}
	
	\begin{adjustbox}{max width=\textwidth}
		\begin{tabular}{c | l | l }
			\hline\hline
			\multicolumn{3}{c}{Specification of dictionary} \\
			\hline
			\multirow{5}{*}{$\mathcal D$} & 1-grams (1-20) & [1], [2], \ldots, [20] \\
			& 2-grams (21 - 30) & [1~2], [2~3], [3~4], [4~5], [5~1], [6~7], [7~8], [8~9], [9~10], [10~6]  \\
			& 2-grams (31 - 40) & [11~12], [12~13], [13~14], [14~15], [15~11], [16~17], [17~18], [18~19], [19~20], [20~11] \\
			& 3-grams (41 - 45) & [11~12~14], [12~13~15], [13~14~12], [14~15~11], [15~11~13]  \\
			& 3-grams (46 - 50) & [1~2~4], [2~3~5], [3~4~7], [3~9~6], [2~5~6] \\
			\hline
		\end{tabular}
	\end{adjustbox}
\end{table}

\begin{table}
	\caption{Simulation results under three simulation settings. }
	\label{simtab}
	\centering
	\begin{adjustbox}{max width=\textwidth}
		\begin{tabular}{c | c c c c c c }
			\hline\hline
			\multicolumn{7}{c}{Setting 1} \\
			\hline
			&  Correct recovery \% & False recovery \% & 2-gram hitting & 3-gram hitting & Class recovery &   \\
			\hline
			No time &  99.6 \% & 0.3 \% & 100.0 \% & 98.8\% & 84 \% & \\
			With time &  99.9 \% & 0.1 \% & 100.0 \% & 99.6\% & 94 \%  &  \\
			\hline\hline
			& & C1 & C2 & C3 & C4 & C5 \\
			\hline
			\multirow{2}{*}{No time} & $\hat \pi$ &  0.40 & 0.298 & 0.199 & 0.053 & 0.049  \\
			& RMSE & 0.015 & 0.014 & 0.014 & 0.007 & 0.007 \\
			\hline
			\multirow{2}{*}{With time} & $\hat \pi$ &  0.399 & 0.299 & 0.202 & 0.049 & 0.051 \\
			& RMSE & 0.016 & 0.013 & 0.014 & 0.006 & 0.007 \\
			\hline\hline
			& & $\lambda_1$ & $\lambda_2$ & $\lambda_3$ & $\lambda_4$ & $\lambda_5$ \\
			\hline
			\multirow{2}{*}{No time} & $\{\hat \lambda_j\}$ &  9.99 & 2.50 & 1.00 & 0.497 & 0.200  \\
			& RMSE & 0.080 & 0.026 & 0.014 & 0.013 & 0.005 \\
			\hline
			\multirow{2}{*}{With time} & $\{\hat \lambda_j\}$ &  10.0 & 2.50 & 0.999 & 0.501 & 0.201  \\
			& RMSE & 0.072 & 0.024 & 0.014 & 0.012 & 0.005 \\
			\hline
		\end{tabular}
	\end{adjustbox}
 
	\medskip 

	\begin{adjustbox}{max width=\textwidth}
	\begin{tabular}{c | c c c c c c c}
		\hline\hline
		\multicolumn{8}{c}{Setting 2} \\
		\hline
		&  Correct recovery \% & False recovery \% & 2-gram hitting & 3-gram hitting & Class recovery & Reduced recovery  & \\
		\hline
		No time &  96.6 \% & 3.5 \% & 100.0 \% & 89.6\% & 0 \% & 92 \% & \\
		With time &  97.3 \% & 2.7 \% & 100.0 \% & 91.8\% & 98 \%  & - & \\
		\hline\hline
		& & C1 & C2 & C3 & C4 & C5 & C6\\
		\hline
		\multirow{2}{*}{No time} & $\hat \pi$ &  0.396 & - & 0.402 & - & 0.102 & 0.100 \\
		& RMSE & 0.018 & - & 0.016 & - & 0.009 & 0.008 \\
		\hline
		\multirow{2}{*}{With time} & $\hat \pi$ &  0.201 & 0.198 & 0.201 & 0.201 & 0.099 & 0.100\\
		& RMSE & 0.012 & 0.014 & 0.012 & 0.012 & 0.008 & 0.009 \\
		\hline\hline
		& & $\lambda_1$ & $\lambda_2$ & $\lambda_3$ & $\lambda_4$ & $\lambda_5$ & $\lambda_6$ \\
		\hline
		\multirow{2}{*}{No time} & $\{\hat \lambda_j\}$ &  0.378 & - & 0.380 & - & 0.997 & 0.999 \\
		& RMSE & 0.017 & - & 0.016 & - & 0.017 & 0.020  \\
		\hline
		\multirow{2}{*}{With time} & $\{\hat \lambda_j\}$ &  0.200 & 4.01 & 0.200 & 0.201 & 1.00 & 1.00\\
		& RMSE & 0.002 & 0.043 & 0.002 & 0.046 & 0.017 & 0.019 \\
		\hline
	\end{tabular}
\end{adjustbox}

	\medskip 
	
\begin{adjustbox}{max width=\textwidth}
	\begin{tabular}{c | c c c c c c }
		\hline\hline
		\multicolumn{7}{c}{Setting 3} \\
		\hline
		&  Correct recovery \% & False recovery \% & 2-gram hitting & 3-gram hitting & 4-gram hitting & Class recovery   \\
		\hline
		With time &  98.9 \% & 1.4 \% & 100.0 \% & 99.0\% & 96.3 \%  & 98 \% \\
		\hline\hline
		& & C1 & C2 & C3 & C4 & C5 \\
		\hline
		\multirow{2}{*}{With time} & $\hat \pi$ &  0.303 & 0.298 & 0.201 & 0.099& 0.100 \\
		& RMSE & 0.010 & 0.010 & 0.009 & 0.007 & 0.007  \\
		\hline\hline
		& & $\lambda_1$ & $\lambda_2$ & $\lambda_3$ & $\lambda_4$ & $\lambda_5$  \\
		\hline
		\multirow{2}{*}{With time} & $\{\hat \lambda_j\}$ &  9.99 & 2.50 & 1.00 & 0.50 & 0.20 \\
		& RMSE & 0.085 & 0.017 & 0.013 & 0.008 & 0.003  \\
		\hline
	\end{tabular}
\end{adjustbox}
\end{table}

\begin{table}
	\caption{Simulation results under the setting 4. }
	\label{simtab:set4}
	\centering
	\begin{adjustbox}{max width=\textwidth}
		\begin{tabular}{c | c c c c  }
			\hline\hline
			\multicolumn{5}{c}{Setting 4} \\
			\hline
			&  $\hat J = 4$  &  $\hat J = 5$ &  &       \\
			\hline
			Percentage &  66 \%  & 30 \% &  &     \\
			\hline\hline
			\hline
			&  Correct recovery \% & False recovery \% & 2-gram hitting & 3-gram hitting      \\
			\hline
			Recovery of $\mathcal D$ &  99.9 \% & 0.1 \% & 100.0 \% & 99.8\%    \\
			\hline\hline
			 & C1 & C2 & C3 & C4  \\
			\hline
			\multirow{1}{*}{$\hat \pi$ ~ (when $\hat J = 4)$} &  0.638& 0.209 & 0.075 & 0.077 \\
			\hline
		\end{tabular}
	\end{adjustbox}
\end{table}

\begin{table}
	\caption{The table contains real data results, including clustering information, most frequent patterns and estimated class-specific parameters.}\label{result}
	\centering
	\scalebox{0.75}{
		\begin{tabular}{l|cccccc}
			\hline
			\hline
			& \multicolumn{6}{|c}{ Summary of six latent classes } \\
			\hline
			&  C1 & C2  &  C3  & C4  & C5  &  C6 \\
			\hline
			$\pi$  & 25.8 \% (0.38 \%)  & 25.0 \% (0.28 \%) & 19.9 \% (0.21 \%)  & 14.8 \% (0.22 \%)  & 8.3 \% (0.18 \%) & 6.0 \% (0.21 \%) \\
			Correct  &  83.4 \% & 89.9 \%  & 98.4 \%  & 66.1 \%  & 54.5 \%  & 84.7 \% \\
			$\{\lambda_j\}$  & 0.52 (4e-3) & 0.82 (6e-3) & 0.87 (5e-3) & 0.69 (5e-3) &  0.62 (7e-3) &  0.27 (4e-3) \\
			Avg Sent.  & 5.9  & 12.8  & 8.1  & 18.4  &  11.7 &  6.1 \\
			Avg Event.  & 17.3  & 34.7  & 23.5  & 46.7  & 30.1  & 18.3  \\
			\hline
			2-grams  &  \multicolumn{6}{|c}{$\theta_{jw}$}  \\
			\hline
			$[20~9]$ & 2e-4  & 6.7e-2  &  \textbf{0.16} & 4.6e-2  & 1.7e-2  & 7.8e-2 \\
			$[10~8]$ &  \textbf{0.10} &  1.5e-2 &  1.7e-2 & 1.5e-2  & 1.5e-2  & 1.7e-2 \\
			$[9~20]$ & 1e-4  & 4.1e-2  & \textbf{0.10}  & 2.8e-2  & 9e-3  & 1.4e-2 \\
			$[8~10]$ & \textbf{7.2e-2}  & 1.4e-2  & 1.6e-2  & 1.2e-2  & 1.8e-2  & 1.0e-2 \\
			$[3~4]$ & 1.0e-2  & 3.8e-2  & 5e-3  & 1.8e-2  & 3.3e-2  & \textbf{8.1e-2} \\
			$[21~14]$ & 7e-3  & 4.5e-2  & 2e-4 & 2.7e-4  & 3.9e-2  & 4.5e-2 \\
			$[14~21]$ & 2.7e-3 &  2.5e-2 &  1e-4 & 1.4e-2  & 2e-2  & 1.9e-2 \\
			$[6~19]$ & 1.3e-2  & 1.0e-2  & 1.4e-2  & 1.8e-2  & 1.2e-2  & 2.9e-2 \\
			$[5~15]$ & 2.0e-2  & 1.6e-2  & 1.5e-2  & 2.0e-2  & 9.6e-3  & 2.7e-2 \\
			$[4~3]$ & 6e-3  & 2.3e-2  & 1.9e-3  & 8e-3  & 1.8e-2  & 2e-2 \\
			$[8~9]$ & 4e-4  & 2.4e-2  &  1.4e-2 & 1.8e-2  & 1.5e-2  & 1.9e-2 \\
			$[6~23]$ &  2e-3 & 1.3e-2 &  2e-4 & 1.1e-2  & 1.4e-2  & 1.9e-2 \\
			$[21~12]$ & 1.5e-3 & 3e-3  & 2e-4  & 6e-3  & 1.7e-2  & 8.7e-3 \\
			$[9~8]$ & 2e-4  & 1.9e-2  & 8e-3  & 1.2e-2  & 9.8e-3  & 1.1e-2 \\
			$[20~8]$ & 1.8e-4  & 4e-3  & 1.9e-3  & 1.2e-2  & 1.7e-2  & 5e-3 \\
			\hline
			3-grams  &   \multicolumn{6}{|c}{$\theta_{jw}$}  \\
			\hline
			$[3~4~22]$ & \textbf{0.11}  & 5.8e-2  & \textbf{0.12}  & 2.2e-2  & 1.4e-2  & 5.1e-2 \\
			$[6~19~16]$  & 6.6e-2  & 6.3e-2  & \textbf{9.7e-2}  & 4.5e-2  & 2.5e-2  & 4.5e-2 \\
			$[10~5~15]$  & 5.8e-2  & 4.1e-2  & \textbf{8.0e-2}  & 2.9e-2  & 1.1e-2  & 7.6e-2 \\
			$[16~19~6]$  & 3.2e-2  & 4.0e-2  & \textbf{6.9e-2}  & 2.3e-2  & 1.3e-2  & 1.6e-2 \\
			$[21~14~22]$  & 7e-3  & \textbf{3.9e-2}  & 1e-2  & 1e-2  & 8.5e-3  & \textbf{5.5e-2} \\
			$[22~4~3]$  & 3.1e-2  & 2.6e-2  & \textbf{6e-2}  & 9.8e-3  & 3.8e-3  & 4e-3 \\
			$[10~8~9]$  & 2.2e-2  & \textbf{5.7e-2}  & \textbf{5.7e-2}  & 2.5e-2  & 2.9e-2  & \textbf{5.0e-2} \\
			$[15~5~10]$ & 1.0e-2  & 2.4e-2  & 4.2e-2  & 1.4e-2  & 3e-3  & 1.1e-2 \\
			$[9~8~10]$  & 3.7e-3  & 3.8e-2  & 4.1e-2  & 1.5e-2  & 1.5e-2  & 1.3e-2 \\
			$[10~8~20]$  & 2e-4  & 3.2e-3  & 1.9e-3  & 7.5e-3  & \textbf{3.3e-2}  & 5.2e-3 \\
			$[3~12~21]$  & 3.7e-3  & 2.3e-3  & 7.5e-4  & 9.1e-3  & \textbf{3.5e-2}  & 3e-3 \\
			$[20~8~10]$  & 4e-4 & 3.4e-3  & 2.1e-3  & 8.5e-3  & \textbf{3.4e-2}  & 6.9e-3 \\
			$[21~12~3]$  & 5.1e-3  & 1.7e-3  & 8e-4  & 7.8e-3 & \textbf{3.2e-2}  & 5.8e-3 \\
			$[22~14~21]$ & 2.1e-3  & 2.6e-2  & 1.0e-2  & 8.3e-3  & 5.1e-3  & 7.8e-3 \\
			$[10~5~7]$  & 1.5e-4  & 2.1e-3  & 3.5e-4  & \textbf{1.3e-2}  & 4.4e-3  & 1.9e-3 \\
			\hline
		\end{tabular}
	}
\end{table}

\begin{table}
	\caption{The table contains the explanation of key patterns.}\label{pattern:explanation}
	\centering
	\scalebox{0.95}{
		\begin{tabular}{l|l}
			\hline
			\hline
			Explanation & Patterns\\
			\hline
			Path connecting ``Silver" and ``Park" & $[20~9]$, $[9~20]$, $[21~14~22]$, $[22~14~21]$  \\
			Path connecting ``Nobel" and ``Park"  & $[10~5~15]$, $[10~8~9]$, $[15~5~10]$, $[9~8~10]$  \\
			Path connecting ``Lincoln" and ``Park"  & $[3~4~22]$, $[6~19~16]$, $[16~19~6]$, $[22~4~3]$  \\
			Path connecting ``Lincoln" and ``Silver" & $[3~12~21]$, $[21~12~3]$  \\
			Path connecting ``Nobel" and ``Silver" &  $[10~8~20]$, $[20~8~10]$ \\
			Partial path between ``Silver" and ``Park" &  $[21~14]$, $[14~21]$   \\
			Partial path between ``Nobel" and ``Park"  & $[10~8]$, $[8~10]$, $[5~15]$, $[8~9]$, $[9~8]$  \\
			Partial path between ``Lincoln" and ``Park"  & $[6~19]$,  $[3~4]$, $[4~3]$, $[]$    \\
			Partial path between ``Lincoln" and ``Silver" & $[21~12]$,   \\
			Partial path between ``Nobel" and ``Silver" & $[20~8]$  \\
			Wrong path & $[6~23]$, $[10,5,7]$ \\
			\hline
		\end{tabular}
	}
\end{table}

\begin{table}
	\caption{Identified patterns from ``Traffic" item. Column ``Class" represents the label of latent class with high corresponding pattern probability. }\label{realdata:dictionary}
	\centering
	\scalebox{0.95}{
		\begin{tabular}{|l|l|}
			\hline
			\hline
			1-grams & Class  \\
			\hline
			$[1]$ & 4 \\
			$[2]$ & 4  \\
			$[3]$ & 5, 6\\
			$[4]$ & 5,6 \\
			$[5]$ & 4,6 \\
			$[6]$ & 6  \\
			$[7]$ & 4 \\
			$[8]$ & 1,4,5,6 \\
			$[9]$ & 1 \\
			$[10]$ & 1,4,5,6\\
			$[11]$ & 4 \\
			$[12]$ & 5\\
			$[13]$ & 4 \\
			$[14]$ & 5, 6 \\
			$[15]$ & 4  \\
			$[16]$ & 4 \\
			$[17]$ & 4 \\
			$[18]$ & 4 \\
			$[19]$ & 4,5,6 \\
			$[20]$ & 1\\
			$[21]$ & 5 \\
			$[22]$ & 2,4,5,6\\
			$[23]$ & 2,4,5,6 \\
			\hline
		\end{tabular}
	
	\begin{tabular}{|l|l|l|l|}
		\hline
		\hline
		 2-grams & Class & 2-grams & Class   \\
		\hline
	    $[20~9]$ & 3 & $[22~9]$ & 5 \\
		 $[10~8]$ & 1 &	$[12~21]$ & 5 \\
		 $[9~20]$ & 3 & $[15~5]$ & 4,6  \\
		 $[8~10]$ & 1 & $[19~6]$ & 6  \\
		 $[3~4]$ & 6  & $[23~6]$ & 2,4,5,6 \\
	     $[21~14]$ & 2,4,5,6 & $[22~4]$ & 1  \\
		$[14~21]$ & 2,4,5,6  & $[12~3]$ & 5 \\
	 $[6~19]$ & 6 & $[5~10]$ & 1,2,3,4 \\
	 $[5~15]$ & 6 & $[10~5]$ & 3,4,6 \\
		  $[4~3]$ & 2,5,6 & $[16~19]$ & 2,4,6 \\
		$[8~9]$ & 2,3,4,5,6 & $[14~22]$ & 2,5,6 \\
		$[6~23]$ & 2,4,5,6  & $[4~14]$ & 5 \\
		 $[21~12]$ & 5 & $[5~7]$ & 4   \\
		 $[9~8]$ & 2  & $[18~17]$ & 1,4 \\
		 $[20~8]$ & 4,5 & $[14~4]$ & 5 \\
		 $[8~20]$ & 4,5 & $[7~5]$ & 4  \\
		 $[3~12]$ & 5  & $[7~11]$ & 4 \\
		 $[9~22]$ & 5 & $[19~11]$ & 4 \\
		 $[19~6]$ & 4,6 &  $[17~18]$ & 1,4 \\
		  $[4~22]$ & 1,6 & & \\
		\hline
	\end{tabular}

		\begin{tabular}{|l|l|}
			\hline
			\hline
			 3-grams &  Class \\
			\hline
			 $[3~4~22]$ & 1,3\\
		    $[6~19~16]$ & 3 \\
			  $[10~5~15]$ & 1,3,6\\
		    $[16~19~6]$ & 3 \\
			 $[21~14~22]$ & 2,6 \\
			$[22~4~3]$ & 1,2,3  \\
			 $[10~8~9]$ & 2,3,6 \\
			 $[15~5~10]$ & 2,3 \\
			$[9~8~10]$ & 2,3 \\
			 $[10~8~20]$ & 5\\
			$[3~12~21]$ & 5 \\
			 $[20~8~10]$ & 5 \\
			 $[21~12~3]$ & 5 \\
			 $[22~14~21]$ & 2 \\
			$[10~5~7]$ & 4 \\
		     $[6~19~11]$ & 4 \\
			 $[7~5~10]$ & 4 \\
			$[11~19~6]$ & 4 \\
			 $[20~10~8]$ & 5,6 \\
			 $[22~3~4]$ & 1,3\\
			\hline
		\end{tabular}
	}
\end{table}

\end{document}